\documentclass[11pt]{article}
\usepackage{fullpage}
\usepackage{comment}
\usepackage{amsmath}
\usepackage{amsthm}
\usepackage{amssymb}
\usepackage{multirow}
\usepackage{url}
\usepackage[final]{showkeys}	
\usepackage{tikz}
\usepackage{lineno}
\usepackage{color}
\usepackage{colortbl}
\usepackage{hhline}
\usepackage{enumerate}
\usepackage{url}
\usepackage{amsmath, amsthm, amssymb, amsfonts, graphicx}
\usepackage{tikz}
\usepackage{cite}
\usepackage{clrscode}
\usepackage[labelfont=bf]{caption}
\usepackage{wrapfig}
\usepackage{subcaption}
\usepackage{bold-extra}
\usepgflibrary{shapes}

\usepackage{algorithmic}
\usepackage[ruled,vlined,commentsnumbered,titlenotnumbered]{algorithm2e}

\newcommand{\threesat}{\textsc{3-SAT}}
\newcommand{\oneinthree}{\textsc{Positive 1-in-3-SAT}}
\newcommand{\nonposoneinthree}{\textsc{1-in-3-SAT}}
\newcommand{\threedm}{\textsc{3-Dimen\-sional Matching}}
\newcommand{\ndm}[1]{\textsc{Numerical $#1$-Dimen\-sional Matching}}
\newcommand{\partition}[1]{\textsc{$#1$-Partition}}
\newcommand{\lengthoffsets}{\textsc{Length Offsets}}
\newcommand{\pathpuzzle}{\textsc{Path Puzzle}}
\newcommand{\problemfont}{}

\newtheorem{theorem}{Theorem}[section]
\newtheorem{lemma}[theorem]{Lemma}

\newtheorem{definition}{Definition}[section]

\newtheorem{defn}[definition]{Definition}
\newtheorem{problem}{Problem}[section]

\usepackage{hyperref}
\usepackage{cleveref} 

\def\be{\begin{enumerate}}
\def\ii{\item}
\def\ee{\end{enumerate}}
\def\bal{\begin{align*}}
\def\eal{\end{align*}}
\def\bi{\begin{itemize}}
\def\ei{\end{itemize}}
\def\bp{\begin{proof}}
\def\ep{\end{proof}}
\def\bl{\begin{lemma}}
\def\el{\end{lemma}}
\def\bt{\begin{thm}}
\def\et{\end{thm}}
\def\bc{\begin{cor}}
\def\ec{\end{cor}}
\def\bd{\begin{defn}}
\def\ed{\end{defn}}
\def\bprop{\begin{prop}}
\def\eprop{\end{prop}}

\usepackage[toc,page]{appendix}

{\makeatletter
 \gdef\xxxmark{%
   \expandafter\ifx\csname @mpargs\endcsname\relax 
     \expandafter\ifx\csname @captype\endcsname\relax 
       \marginpar{xxx}
     \else
       xxx 
     \fi
   \else
     xxx 
   \fi}
 \gdef\xxx{\@ifnextchar[\xxx@lab\xxx@nolab}
 \long\gdef\xxx@lab[#1]#2{{\bf [\xxxmark #2 ---{\sc #1}]}}
 \long\gdef\xxx@nolab#1{{\bf [\xxxmark #1]}}
}

\makeatletter
\newcommand{\removelatexerror}{\let\@latex@error\@gobble}
\makeatother

\begin{document}

\def \isnotin {\nsubseteq}

\def \eps {\varepsilon}

\title{Path Puzzles: Discrete Tomography with a Path Constraint is Hard}
\author{
  Jeffrey Bosboom%
    \thanks{Massachusetts Institute of Technology, \protect\url{{jbosboom,edemaine,mdemaine,achester,jkopin}@mit.edu}}
\and
  Erik D. Demaine\footnotemark[1]
\and
  Martin L. Demaine\footnotemark[1]
\and
  Adam Hesterberg\footnotemark[1]
\and
  Roderick Kimball%
    \thanks{Enigami Puzzles \& Games}
\and
  Justin Kopinsky\footnotemark[1]
}
\date{}
\maketitle

\begin{abstract}
  We prove that path puzzles with complete row and column information---or
  equivalently, 2D orthogonal discrete tomography with Hamiltonicity
  constraint---are strongly NP-complete, ASP-complete, and \#P-complete.
  Along the way, we newly establish ASP-completeness and \#P-completeness for
  \threedm\ and \ndm{3}.
\end{abstract}

\section{Introduction}

Path puzzles are a type of pencil-and-paper logic puzzle introduced in
Roderick Kimball's 2013 book~\cite{Kimball2013} and featured in
\emph{The New York Times}'s Wordplay blog \cite{NYTimes}.
Figure~\ref{fig:example} gives a small example.
A puzzle consists of a (rectangular) grid of cells with two exits
(or ``doors'') on the boundary and numerical constraints
on some subset of the rows and columns. A solution consists of a single
non-intersecting path which starts and ends at two boundary doors and
which passes through a number of cells in each constrained row and column
equal to the given numerical clue.
Many variations of path puzzles are given in~\cite{Kimball2013} and
elsewhere, for example using non-rectangular grids, grid-internal
constraints, and additional candidate doors,
but these generalizations make the problem only harder.

\begin{figure}
\centering
\includegraphics[width=.35\textwidth]{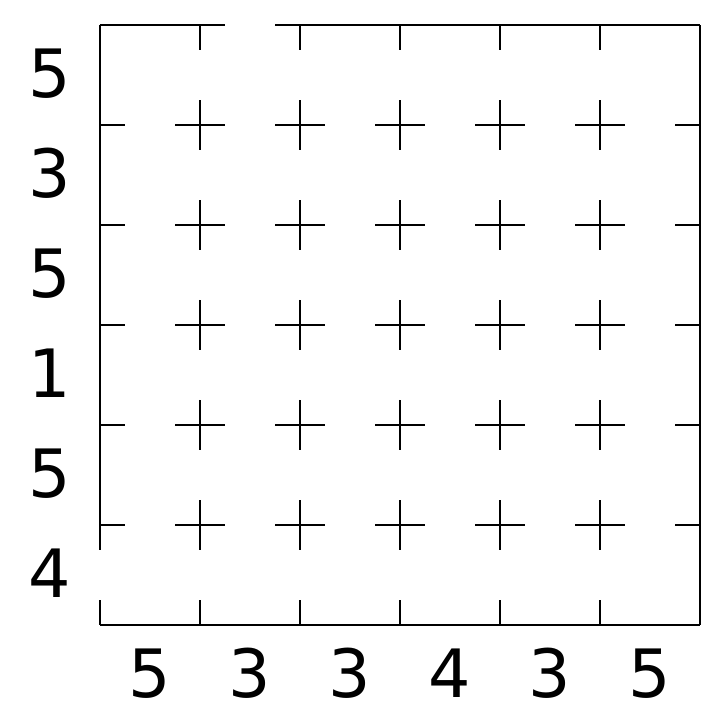}
\hfil
\includegraphics[width=.35\textwidth]{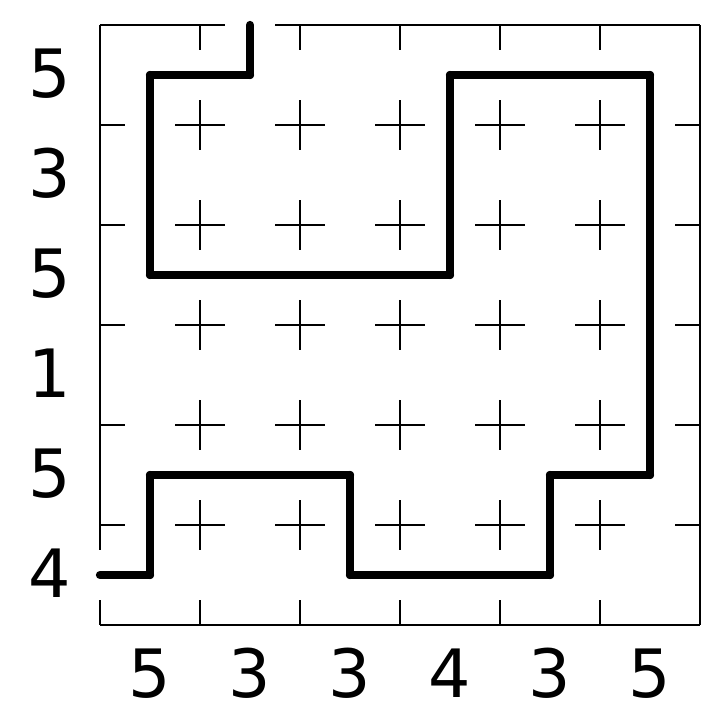}
\caption{A \pathpuzzle{} with complete row/column information (left) and its solution (right).}
\label{fig:example}
\end{figure}

Path puzzles are closely related to
\emph{discrete tomography}~\cite{herman2012discrete}, in particular
the 2D orthogonal form: given the number of black pixels in each row
and column, reconstruct a black-and-white image.  This problem arises
naturally in reconstruction of shapes via x-ray images (which measure density).
Vanilla 2-dimensional discrete tomography is known to
have efficient (polynomial-time) algorithms~\cite{herman2012discrete}, though
it becomes hard under certain connectivity constraints on the output
image~\cite{LungoNivat}.%
\footnote{Most sets of row and column constraints are ambiguous;
constraining the output image makes the problem harder by preventing an easy
image from being found instead.}
A path puzzle is essentially a 2-dimensional discrete tomography
problem with partial information (not all row and column counts) and
an additional Hamiltonicity (single-path) constraint on the output image.

\paragraph{Our results.}
Unlike 2-dimensional discrete tomography, we show that path puzzles
are NP-complete, even with perfect information
(i.e., with all row and column counts specified).
In other words, 2-dimensional discrete tomography becomes NP-complete
with an added Hamiltonicity constraint.
In fact, we prove the stronger results that perfect-information
path puzzles are \textsc{Another Solution Problem} (ASP) hard and
(to count solutions) \#P-complete.

Figure~\ref{fig:chain} shows the chain of reductions we use to prove
hardness of \pathpuzzle.
To preserve hardness for the ASP and \#P classes, our reductions are
\emph{parsimonious}; that is, they preserve the number of solutions
between the source and target problem instances, generally by showing a
one-to-one correspondence thereof.
We start from the source problem of \oneinthree{} which is known to be
ASP-hard \cite{Seta02thecomplexities,Hunt-Marathe-Radhakrishnan-Stearns-1998}
and (to count solutions)
\#P-complete~\cite{Hunt-Marathe-Radhakrishnan-Stearns-1998}.
Along the way, we newly establish strong ASP-hardness and \#P-completeness for
\threedm{}, \ndm{4}, \ndm{3}, and a new problem \lengthoffsets{},
in addition to \pathpuzzle.

\begin{figure*}
\centering
\footnotesize
\usetikzlibrary{arrows,positioning,chains,graphs,matrix,quotes}
\definecolor{mycolor1}{rgb}{0.757, 0.690, 0.867}
\tikzset{
        problem/.style={rectangle, rounded corners=4pt, draw, line width=1.6pt, text width=2.2cm, minimum height=1cm, text centered, fill=mycolor1!50}, 
        problemNarrow/.style={rectangle, rounded corners=4pt, draw, line width=1.6pt, text width=1.3cm, minimum height=1cm, text centered, fill=mycolor1!50}, 
        problemMedium/.style={rectangle, rounded corners=4pt, draw, line width=1.6pt, text width=1.7cm, minimum height=1cm, text centered, fill=mycolor1!50}, 
        hv path/.style = {to path={-| (\tikztotarget) \tikztonodes}},
        vh path/.style = {to path={|- (\tikztotarget) \tikztonodes}},
        every path/.append style={line width=2pt},
}
\begin{tikzpicture}
\matrix[row sep = 5mm, column sep = 8.5mm]{
        \node (1in3) [problemMedium] {\oneinthree};
&
        \node (3DM) [problemMedium] {\threedm};
&
        \node (N4DM) [problem] {\ndm4};
&
        \node (N3DM) [problem] {\ndm3};
&
        \node (LO) [problemNarrow] {\lengthoffsets};
&
        \node (PP) [problemNarrow] {\pathpuzzle};
\\
};
\def\stack#1#2{\vbox{\setbox0=\hbox{#1}\copy0\hbox to \wd0{\hfil #2\hfil}\vskip 4pt}}
\def\thmref#1{\textcolor{darkgray}{\stack{Thm}{\ref{#1}}}}
\graph[use existing nodes]{
  1in3 ->["\thmref{thm:3dm}"]
  3DM ->["\thmref{thm:n4dm}"]
  N4DM ->["\thmref{thm:n3dm}"]
  N3DM ->["\thmref{thm:length-offsets}"]
  LO ->["\thmref{thm:lo-to-pp}"]
  PP};
\end{tikzpicture}
\caption{The chain of reductions used in our proof.}
\label{fig:chain}
\end{figure*}
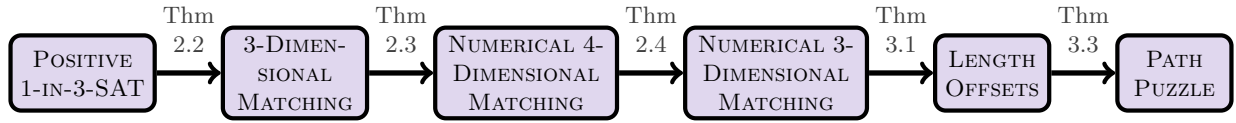

\paragraph{Fonts.}
To further communicate the challenge of path puzzles to the general public,
we designed a mathematical puzzle typeface
(as part of a series\footnote{See \url{http://erikdemaine.org/fonts}}).
Figure~\ref{fig:font-unsolved} gives the puzzle font, which has
one path puzzle for each letter of the alphabet.
Their solutions are designed to look like the 26 letters of the alphabet,
and are verified unique by exhaustive search.
Look ahead to the solved font in Figure~\ref{fig:font-solved} in 
Appendix~\ref{app:font-sol} when you no longer want to solve the puzzles.

\begin{figure}
  \centering
  \includegraphics[width=.13\linewidth]{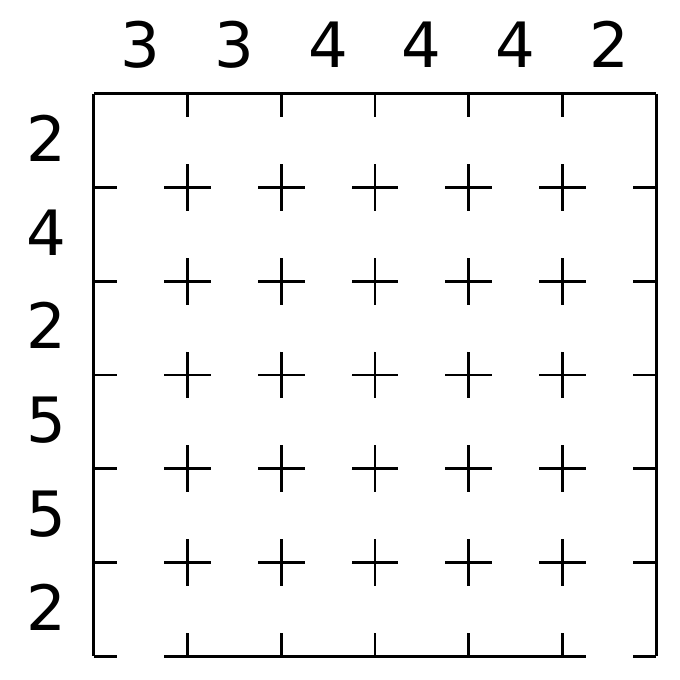}
  \includegraphics[width=.13\linewidth]{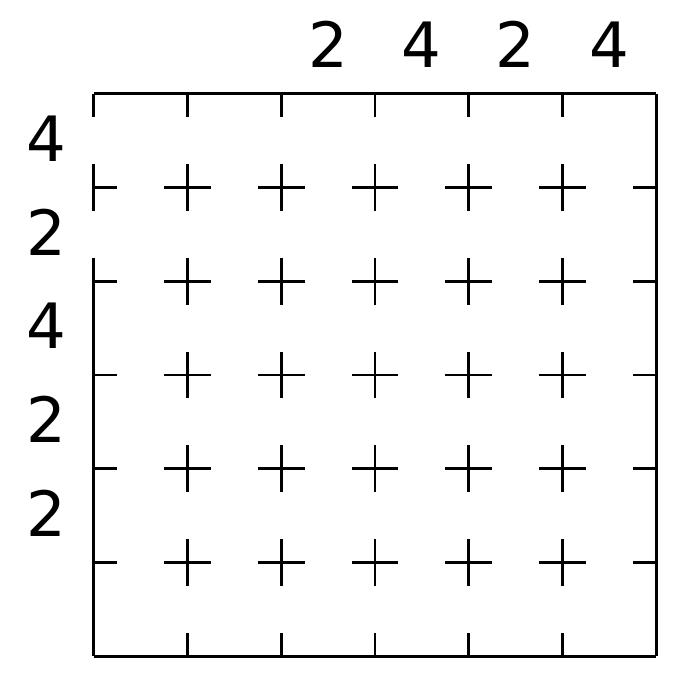}
  \includegraphics[width=.13\linewidth]{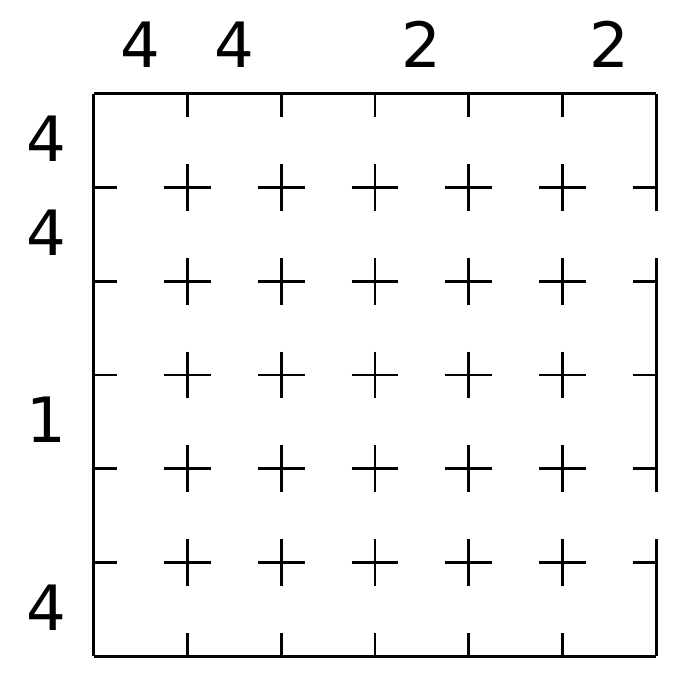}
  \includegraphics[width=.13\linewidth]{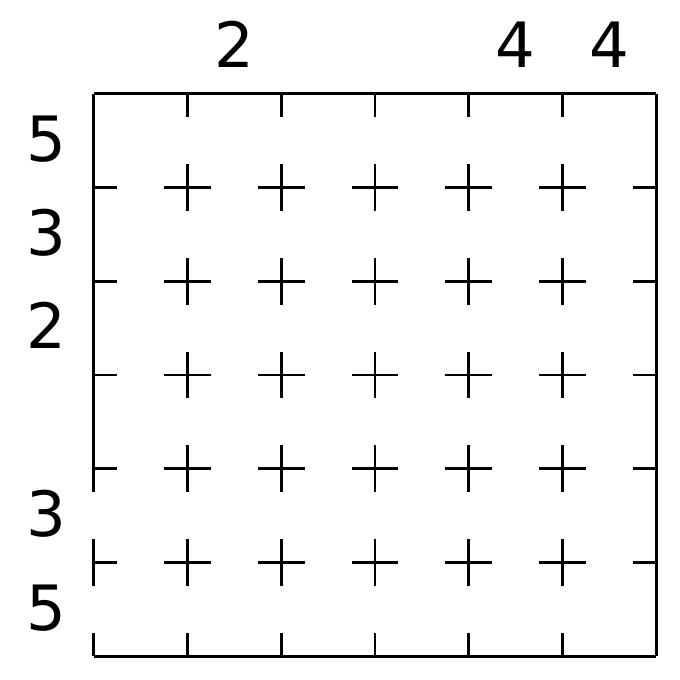}
  \includegraphics[width=.13\linewidth]{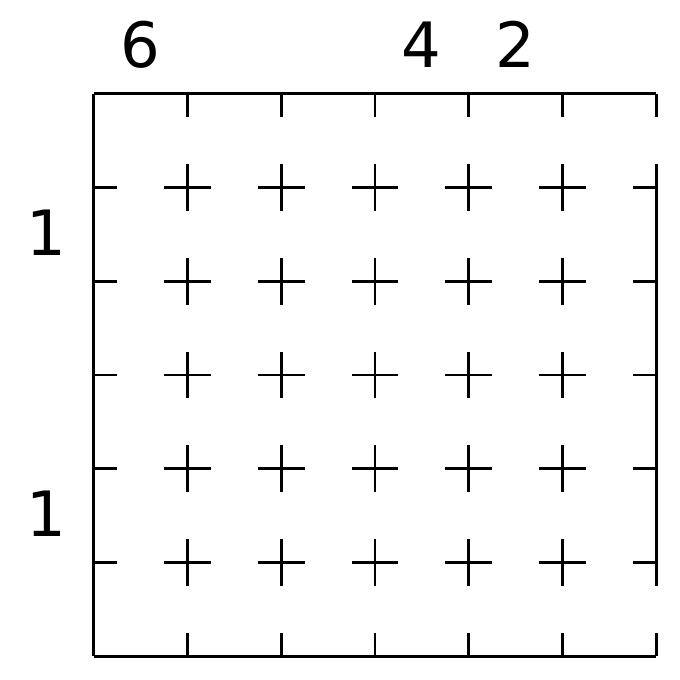}
  \includegraphics[width=.13\linewidth]{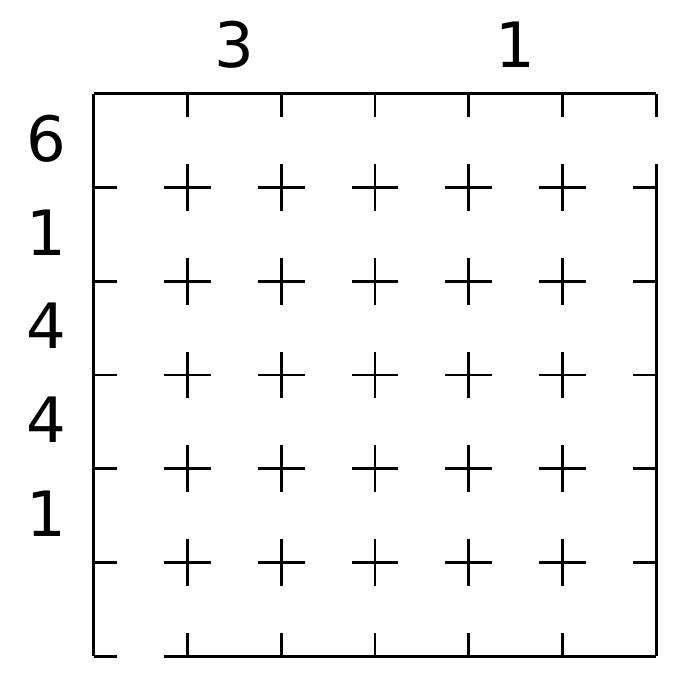}
  \includegraphics[width=.13\linewidth]{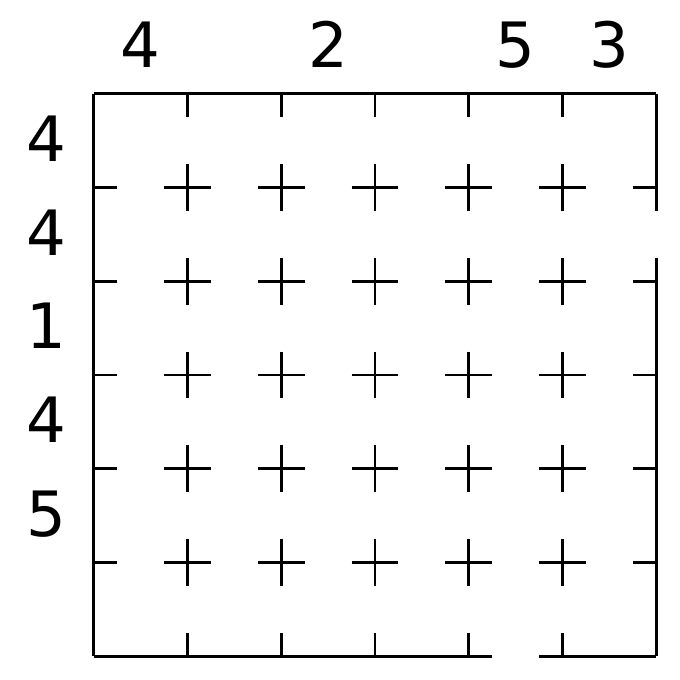}

  \includegraphics[width=.13\linewidth]{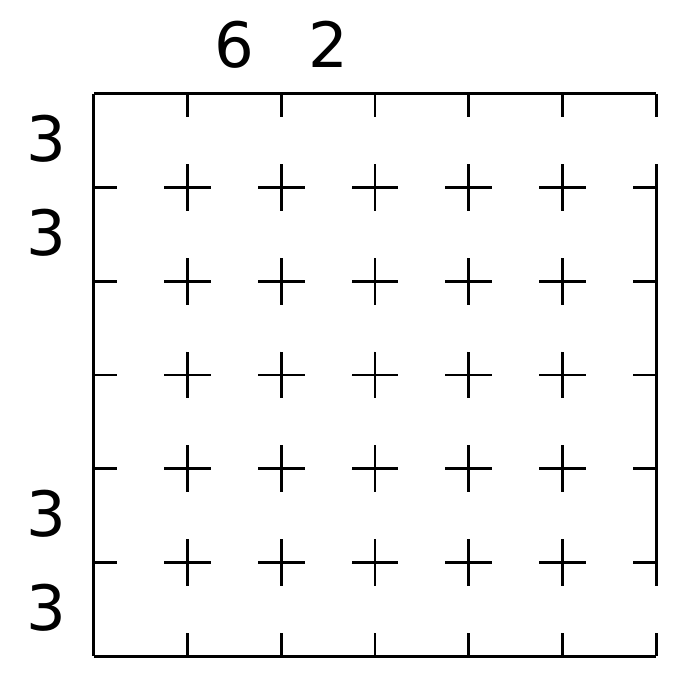}
  \includegraphics[width=.13\linewidth]{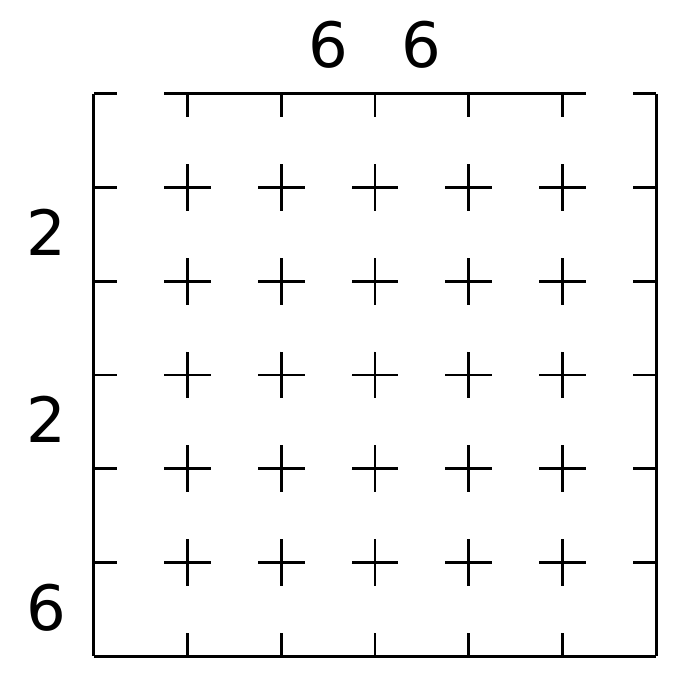}
  \includegraphics[width=.13\linewidth]{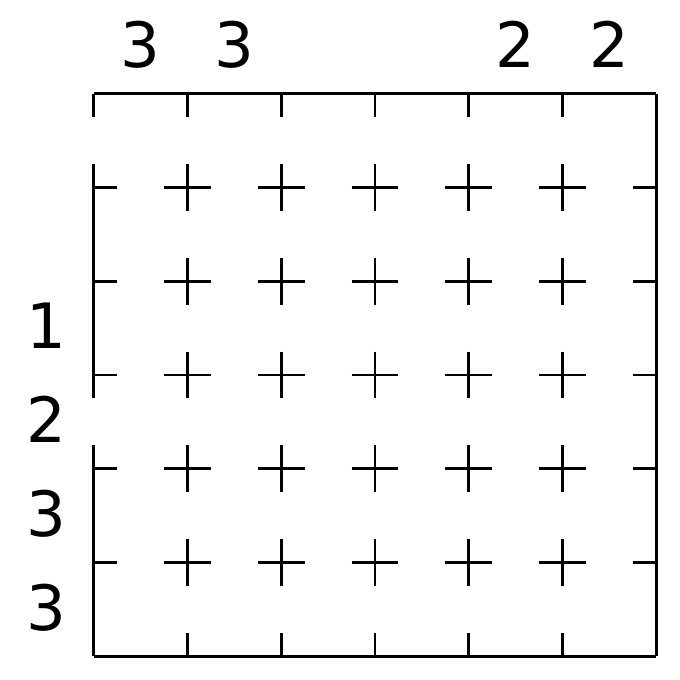}
  \includegraphics[width=.13\linewidth]{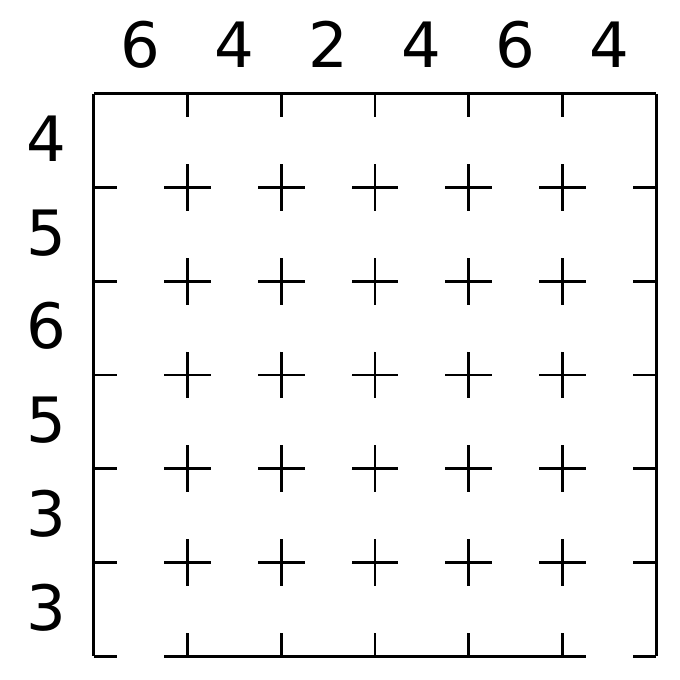}
  \includegraphics[width=.13\linewidth]{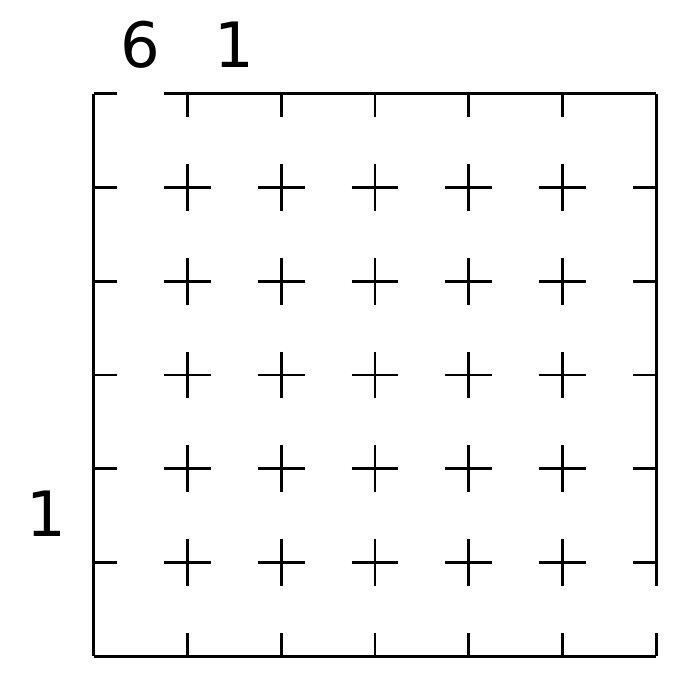}
  \includegraphics[width=.13\linewidth]{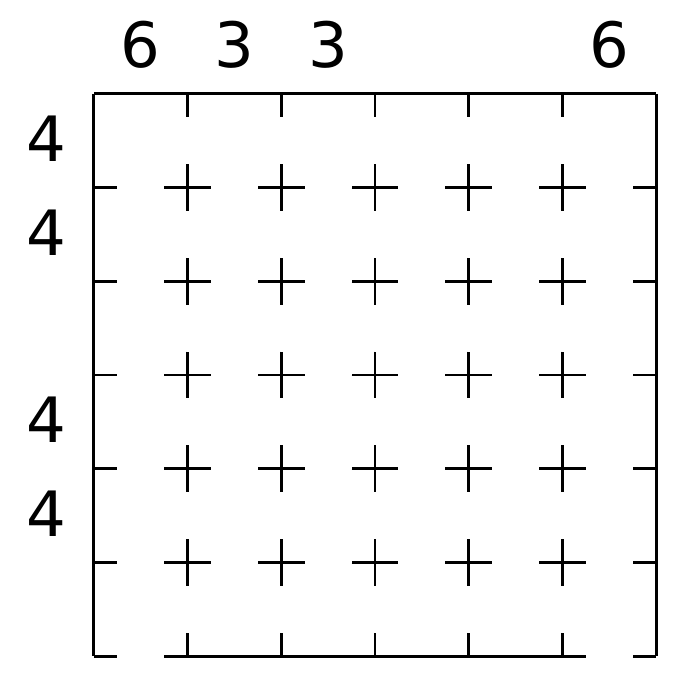}
  \includegraphics[width=.13\linewidth]{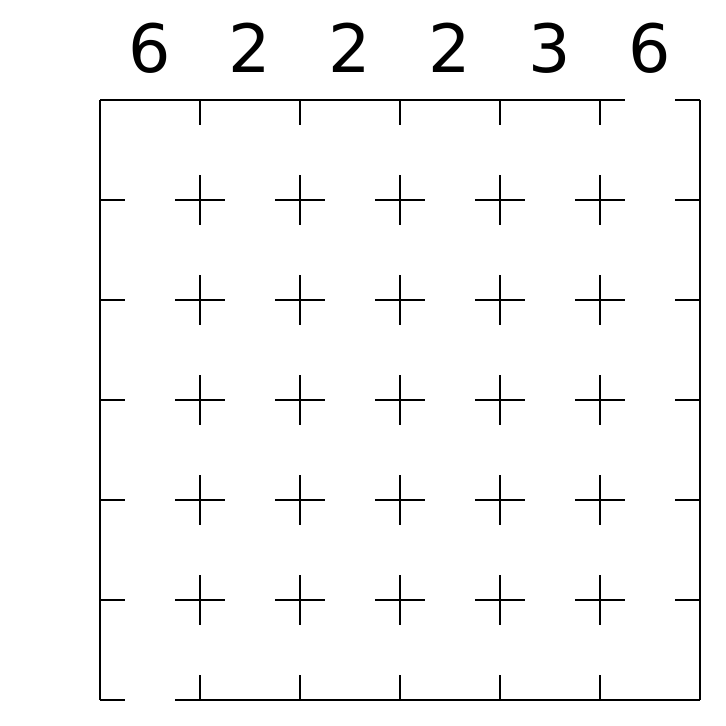}

  \includegraphics[width=.13\linewidth]{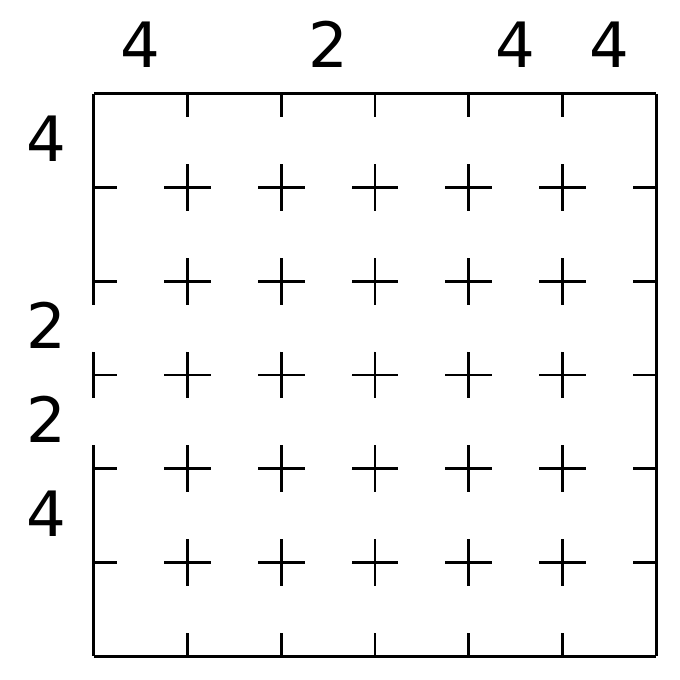}
  \includegraphics[width=.13\linewidth]{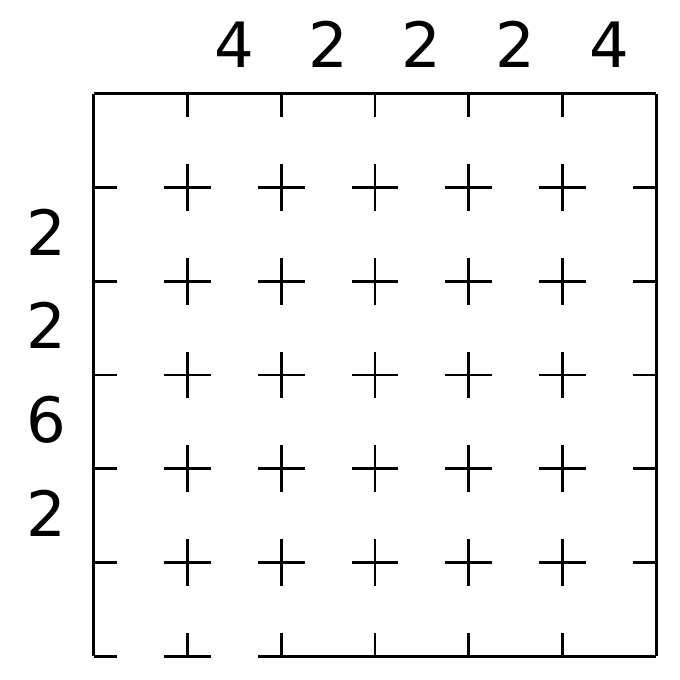}
  \includegraphics[width=.13\linewidth]{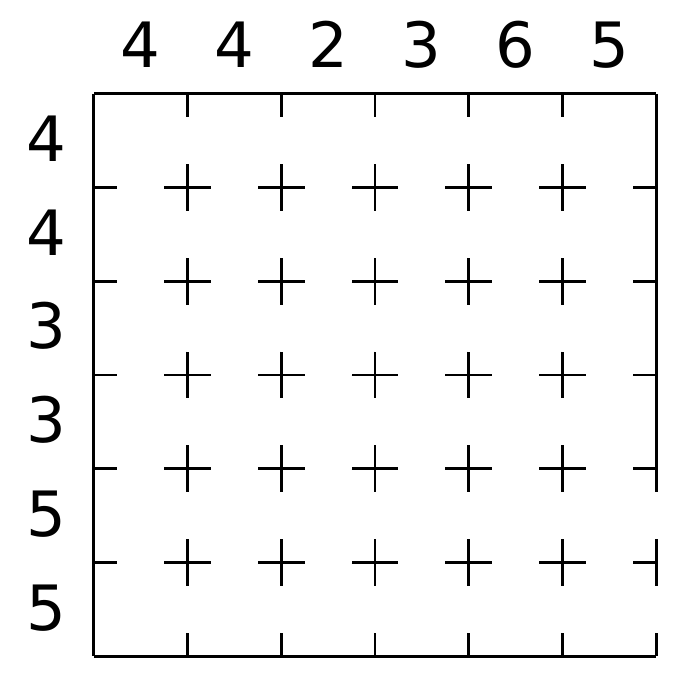}
  \includegraphics[width=.13\linewidth]{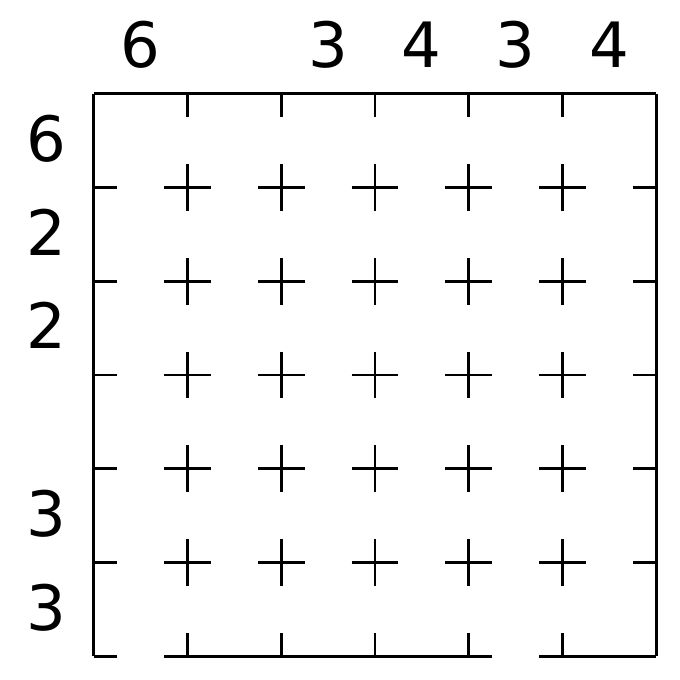}
  \includegraphics[width=.13\linewidth]{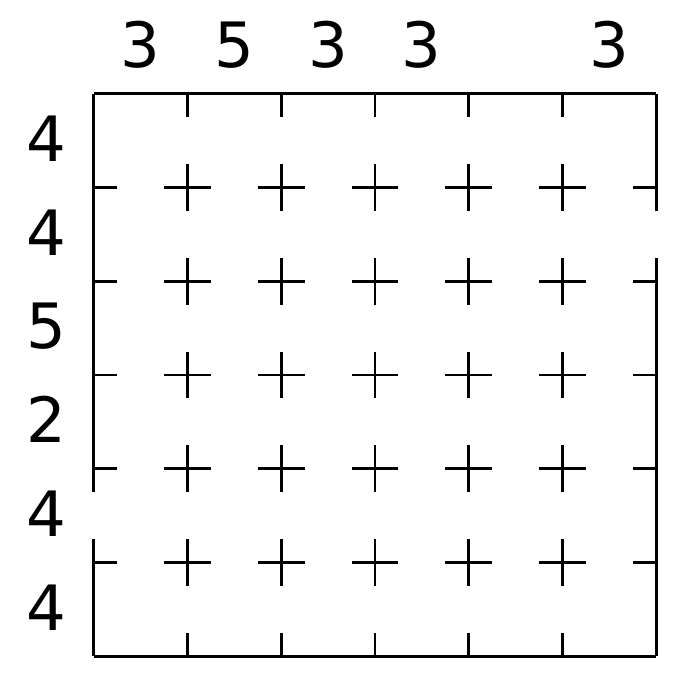}
  \includegraphics[width=.13\linewidth]{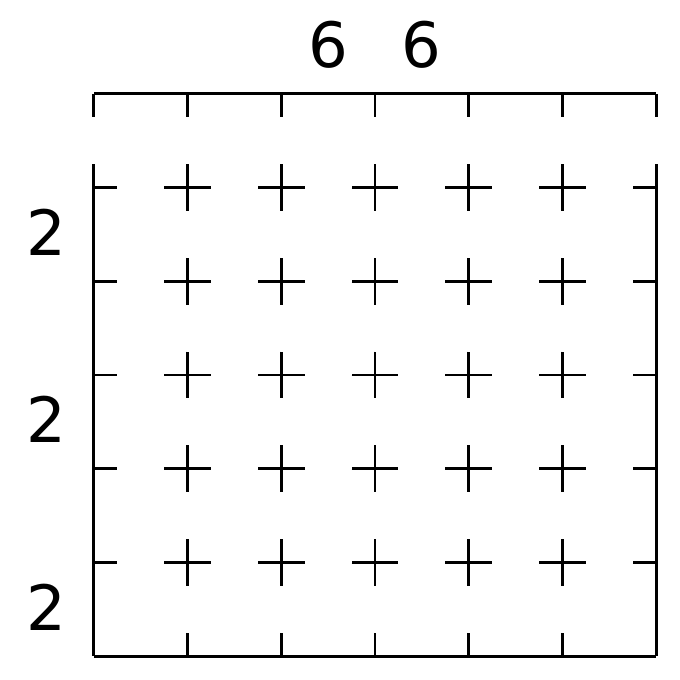}
  \hspace{.13\linewidth}

  \includegraphics[width=.13\linewidth]{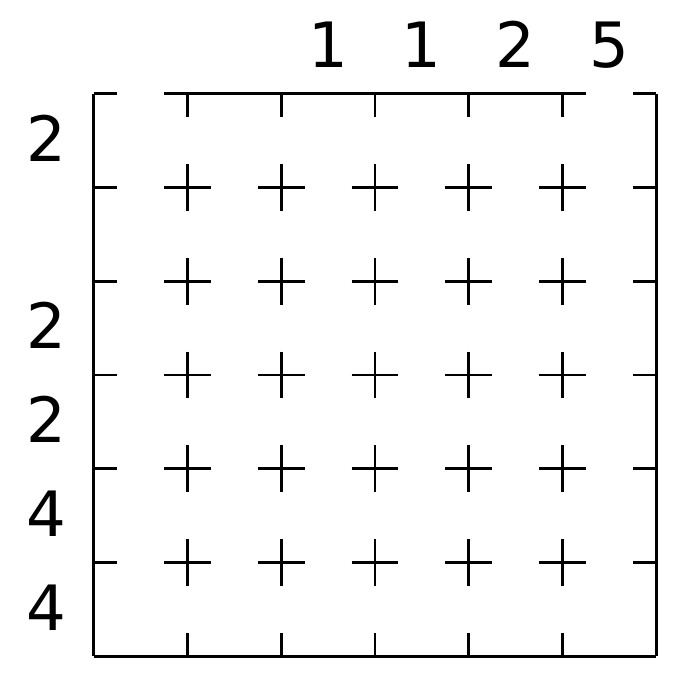}
  \includegraphics[width=.13\linewidth]{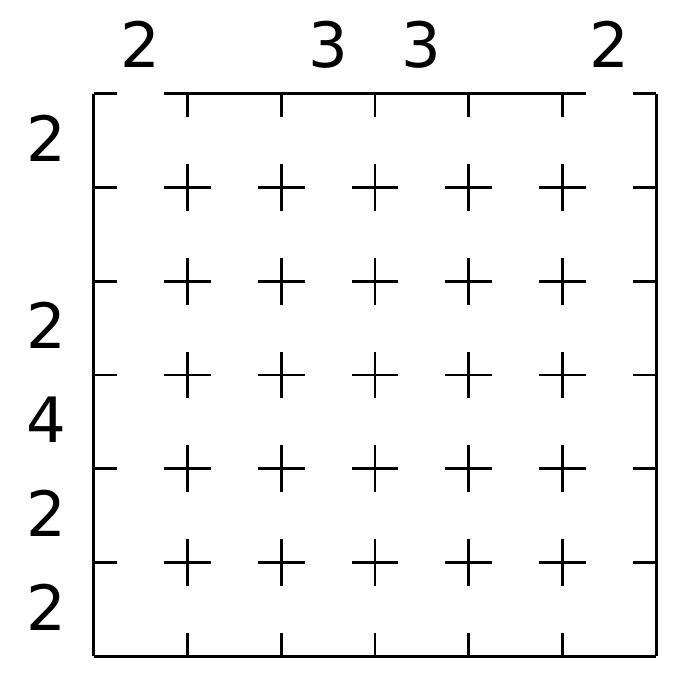}
  \includegraphics[width=.13\linewidth]{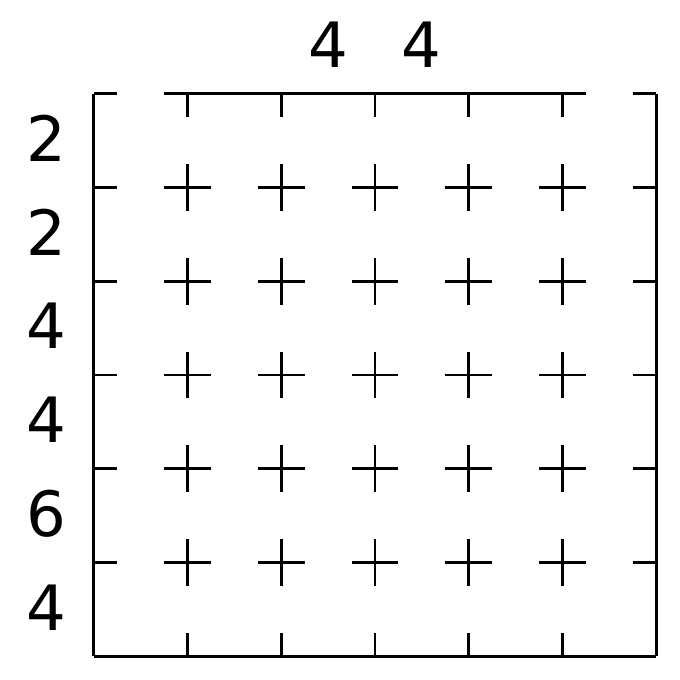}
  \includegraphics[width=.13\linewidth]{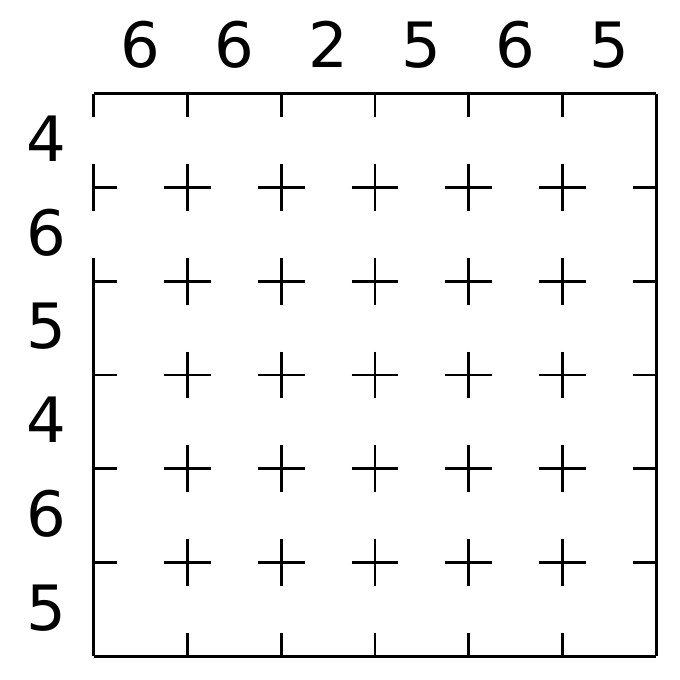}
  \includegraphics[width=.13\linewidth]{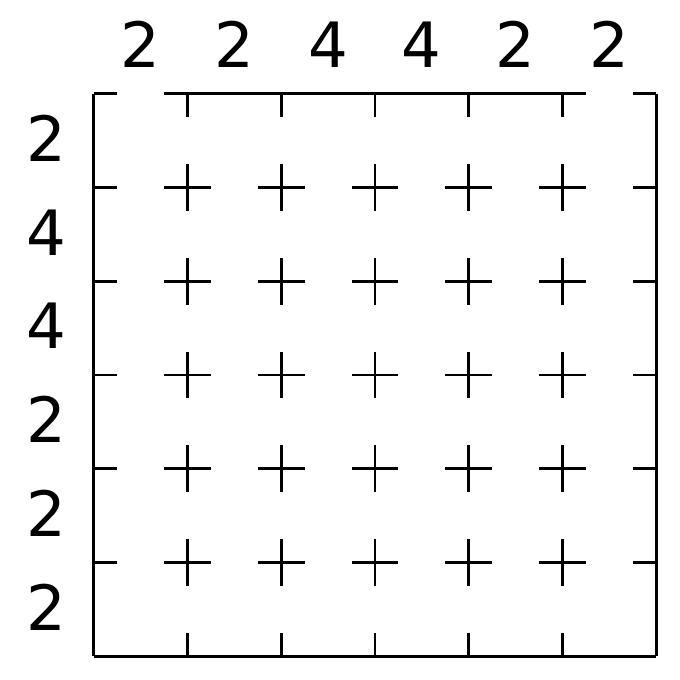}
  \includegraphics[width=.13\linewidth]{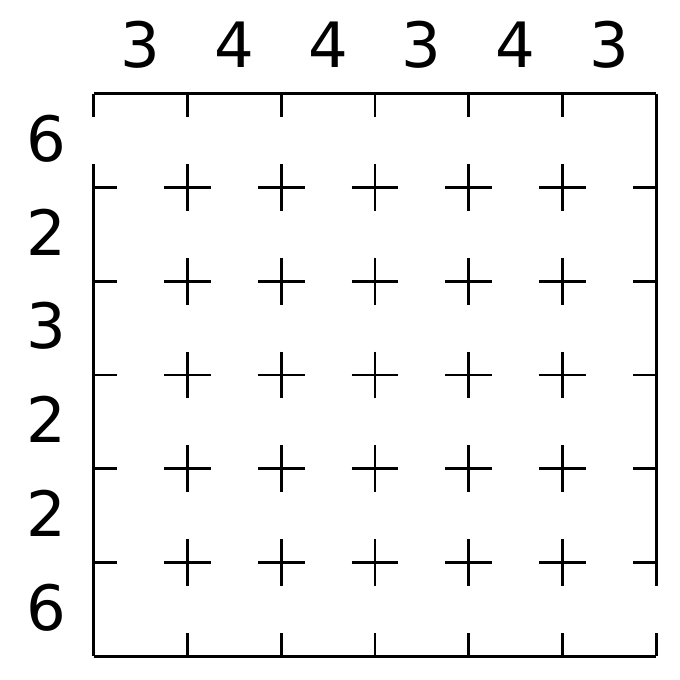}
  \hspace{.13\linewidth}
  \caption{Puzzle font}
  \label{fig:font-unsolved}
\end{figure}

\section{Numerical 3DM is ASP-Complete and \#P-Complete}
\label{ndm}
The goal of this section is to prove that \ndm3 is \emph{strongly}
ASP- and \#P-complete, i.e., ASP- and \#P-complete even when the $n$ numbers
are bounded by a polynomial in~$n$.
We follow a similar chain of reductions by Garey and Johnson
\cite{Garey-Johnson-1979}, namely 3SAT $\to$ \threedm\ $\to$ \partition4
$\to$ \partition3, but replacing \partition{k} with \ndm{k}
and starting from a different version of SAT:

\begin{problem}[\problemfont\oneinthree]
Given a 3CNF formula $C$ with only positive literals, is there an assignment of variables such that exactly one literal in each clause of $C$ is true?
\end{problem}

\begin{lemma}
\label{lem:1in3}
\oneinthree{} is ASP-hard and \#P-hard.
\end{lemma}

\begin{proof}
\textsc{3SAT} is shown to be \#P-hard in~\cite{valiant1979}.
Section~3.2.1 of~\cite{Seta02thecomplexities} shows that
\textsc{3SAT} is ASP-hard.%
\footnote{Section 3.2.4 of~\cite{Seta02thecomplexities} proves that
  \nonposoneinthree\ is ASP-hard.  Unfortunately, their problem definition
  allows negative clauses, while we need \oneinthree.}
Theorem~3.8 of~\cite{Hunt-Marathe-Radhakrishnan-Stearns-1998}
gives a parsimonious reduction from \textsc{3SAT} to \oneinthree{}.%
\footnote{In \cite{Hunt-Marathe-Radhakrishnan-Stearns-1998},
  \oneinthree\ is called ``\textsc{1-Ex3MonoSat}''.}
Combining these results gives the claim.
\end{proof}


%

\begin{problem}[\problemfont\threedm]
Given three sets $X, Y, Z$ of equal cardinality and a set $T$ of triples $(x,y,z)$ where $x\in X, y\in Y, z\in Z$, is there a set $S \subseteq T$ such that each element of $X, Y, Z$ appears in exactly one triple in~$S$?
\end{problem}

\begin{theorem} \label{thm:3dm}
\threedm{} is ASP-hard and \#P-hard, even when $T$ is constrained not to contain any two triples agreeing on more than one coordinate.
\end{theorem}
\begin{proof}
We give a parsimonious reduction from \oneinthree{},
using the variable gadget from Garey and Johnson's reduction
from \threesat{} to \threedm{} \cite[Thm.~3.2, p.~50]{Garey-Johnson-1979}.
Given a \oneinthree{} instance with a set $V$ of variables and $C$ of clauses,
we construct the corresponding \threedm{} instance as follows. 
We will represent the \threedm{} instance as a hypergraph that is tripartite
and 3-uniform, i.e., in which each edge connects exactly three vertices of
different colors according to a 3-coloring of the vertices.
We will say that vertices colored $0$ belong to~$X$,
vertices colored $1$ belong to~$Y$,
and vertices colored $2$ belong to~$Z$.

\paragraph{Clause triplication.}
First we triplicate each clause, producing the multiset
$C' = C \sqcup C \sqcup C$ (the disjoint union of three copies of~$C$).
As a result, the number $n_x$ of occurrences of each variable $x \in V$
in clauses in $C'$ (multiply counting if $x$ occurs multiple times in the
same clause) is divisible by~$3$.
A truth assignment for $V$ satisfies $C'$ if and only if it satisfies $C$,
so this triplication does not affect correctness,
but it will help us obtain a 3-coloring.

\paragraph{Variable gadget.}
Next, for each variable $x \in V$, we create a \emph{variable gadget}
consisting of $4 n_x$ vertices associated with~$x$;
refer to Figure~\ref{fig:3dm}.
We call $n_x$ of the vertices \emph{positive} $x$ vertices,
denoted $x_0, x_1, \dots, x_{n_x-1}$ (one for each occurrence of $x$ in~$C'$);
we call $n_x$ of the vertices \emph{negative} $x$ vertices,
denoted $\bar x_0, \bar x_1, \dots, \bar x_{n_x-1}$;
and we call $2 n_x$ of them \emph{auxiliary} vertices,
denoted $x'_0, x'_1, \dots, x'_{2 n_x-1}$.
The edges covering the auxiliary vertices are as follows:
for each $i \in \{0, 1, \dots, n_x-1\}$,
we add the ``positive'' edge $(x_i, x'_{2i},  x'_{2i+1})$
and the ``negative'' edge $(\bar{x}_i,  x'_{2i+1}, x'_{(2i+2) \bmod 2 n_x})$.
No other edges cover the auxiliary vertices,
so there are only two ways to cover them: choose all the positive edges,
thereby covering all the positive vertices and none of the negative vertices,
or choose all the negative edges,
thereby covering all the negative vertices and none of the positive vertices.
The former choice represents assigning $x$ to be \textsc{true},
while the latter choice represents assigning $x$ to be \textsc{false}.

\begin{figure}
\centering
\begin{minipage}[t]{0.49\textwidth}
  \centering
  \includegraphics[scale=0.4]{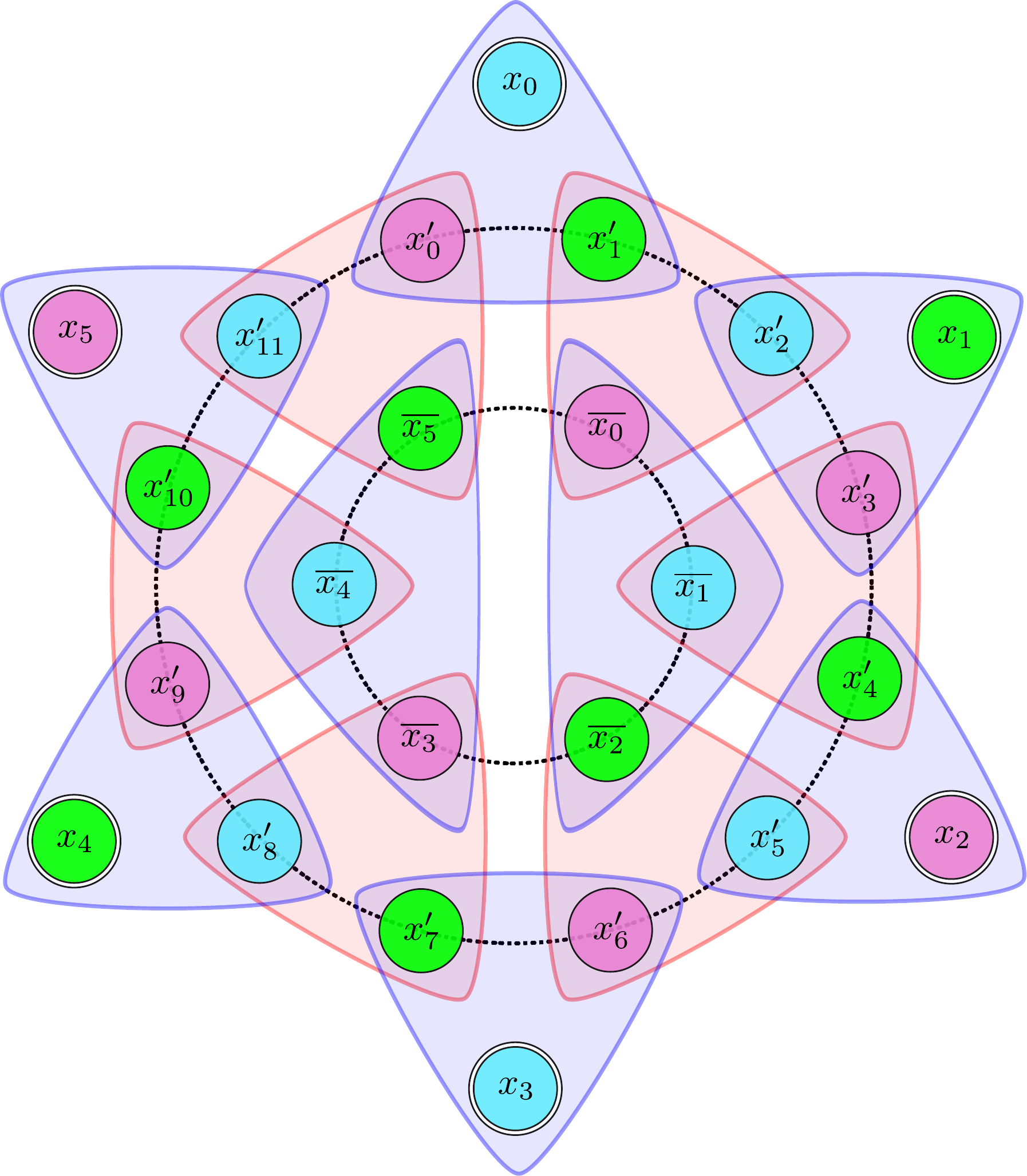}
  \caption{3DM variable gadget for variable $x$ occurring in $n_x=6$ clauses in $C'$ (two clauses in~$C$). Hyperedges are drawn as shaded triangles; any solution must include all the positive (blue) or all the negative (red) hyperedges. Vertex colors $0,1,2$ are drawn as magenta, green, and cyan. Only the positive vertices (drawn with doubled outlines) are attached to other gadgets.}
  \label{fig:3dm}
\end{minipage}\hfill
\begin{minipage}[t]{0.49\textwidth}
  \centering
  \hbox to \hsize{\hss\includegraphics[scale=0.4]{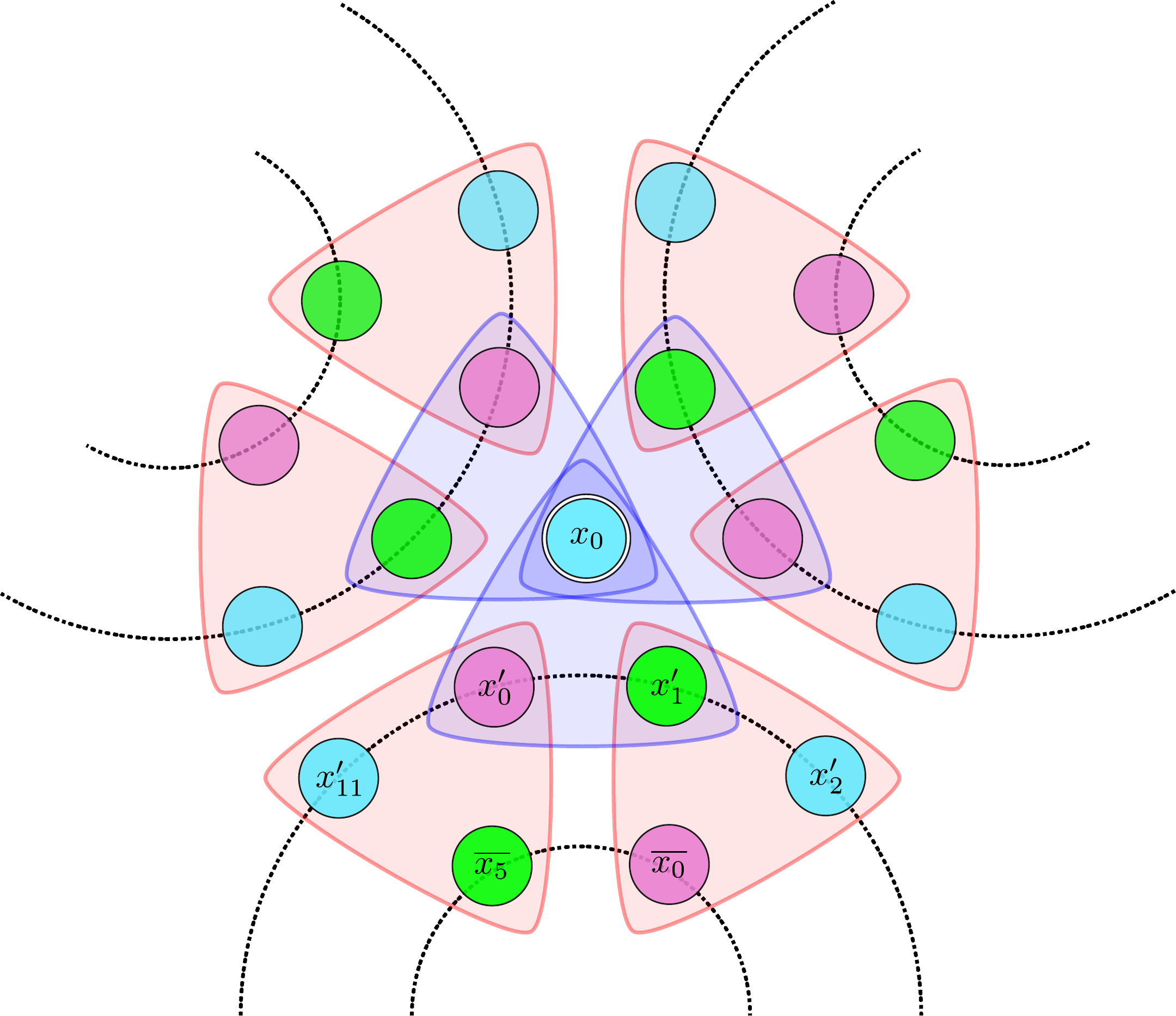}}
  \caption{3DM clause gadget, identifying cyan positive vertices from three variable gadgets (Figure~\ref{fig:3dm}). Although here we draw the three variables as distinct, we may also identify positive vertices from the same variable gadget (when the same variable appears twice in the same clause).}
  \label{fig:3dm-clause}
\end{minipage}
\end{figure}

Because $n_x$ is divisible by $3$, we can 3-color the vertices by assigning
colors $0,1,2,0,1,2,\dots$ to the auxiliary vertices $x'_0,x'_1,x'_2,\dots$,
then coloring each positive and negative vertex with the one color not used in
its edge, as shown in Figure~\ref{fig:3dm}. 

The final edges of the variable gadget serve to collect ``garbage''
negative vertices.
For each $i \in \{0, 1, \dots, n_x/3-1\}$
(using that $n_x$ is divisible by~$3$), we add another positive edge
$(\bar x_{3 i}, \bar x_{3 i+1}, \bar x_{3 i+2})$.
These positive edges overlap the negative edges, so cannot be chosen in the
\textsc{false} assignment, but do not overlap the positive edges,
and then they cover all the negative vertices.
No other edges cover the negative vertices, so again we must choose all
the positive edges or all the negative edges.

Therefore,
local to the variable gadget,
we cover all the auxiliary and negative $x$ vertices, and either all or none
of the positive $x$ vertices.  Only the positive $x$ vertices will interact
with other gadgets, through clause gadgets.

\paragraph{Clause gadget.}
Finally, for each clause $c = \langle x,y,z \rangle \in C'$, we
\emph{identify} one positive $x$ vertex, one positive $y$ vertex, and
one positive $z$ vertex all of the same color,
resulting in a single vertex covered by one positive edge
from each of the three corresponding vertex gadgets;
refer to Figure~\ref{fig:3dm-clause}.
The three identified vertices are chosen to be unique to this clause gadget,
so they will not be identified again, and thus will be covered exactly once
if and only if exactly one of the three variables is assigned \textsc{true}.

For each of the three copies of a clause in~$C$, we choose the identified
vertices to be a different color among $\{0,1,2\}$, so that each clause in $C$
consumes exactly one positive vertex of each color from each of the three
variable gadgets.
(When a variable appears twice in the same clause, two of these variable
gadgets will actually be the same, and we will end up consuming two positive
vertices of each color, but the accounting remains the same.)
Thus we will be able to use each positive vertex in each variable gadget
exactly once, without running out of any particular color.

\paragraph{Equivalence.}
Because the identified $(x,y,z)$ vertex in a clause gadget must be covered by
exactly one edge in the \threedm{} problem, exactly one of $x,y,z$ must have
an assignment of \textsc{true}, which is the \oneinthree\ constraint.
Thus, given a solved instance of \threedm{}, we can extract exactly one
solution to the original \oneinthree{} instance. 
Furthermore, given a solution to the \oneinthree{} instance, we can produce
exactly one solution to the \threedm{} instance by choosing all the positive
$x$ edges if $x$ is set to \textsc{true} and all the negative edges if $x$ is
set to \textsc{false}.
Thus the reduction is parsimonious.

Examining Figures~\ref{fig:3dm} and~\ref{fig:3dm-clause},
we also see that no hyperedge shares more than one vertex, as claimed.
\end{proof}

\begin{problem}[\problemfont\ndm{k}]
Given $k$ multisets of positive integers $X_1, \dots, X_k$ and a positive integer target sum $t$, does there exist a set $S \subseteq X_1\times \dots \times X_k$ of $k$-tuples such that, for each $(x_1, \dots, x_k) \in S$, $x_1 + \dots + x_k = t$, and each element of each $X_i$ appears as the $i$th coordinate in exactly one element of~$S$?  (Thus $|X_1| = \dots = |X_k| = |S|$, and we denote this common size by~$n$.)
\end{problem}

We will consider specially the cases $k=3$ and $k=4$
for which we label the sets $X,Y,Z$ and $W,X,Y,Z$ respectively.

\begin{theorem} \label{thm:n4dm}
  \ndm{4} is strongly ASP-hard and \#P-hard, even if $Y \cup (Y+Z)$
  (where $Y+Z=\{y+z : y\in Y, z\in Z\}$)
  is guaranteed to be a set (not a multiset). 
\end{theorem}

\begin{proof}
We give a parsimonious reduction from \threedm{} where no two triples in
$T$ agree on more than one coordinate, as guaranteed by Theorem~\ref{thm:3dm}.
Our reduction loosely follows Garey and Johnson's original reduction \cite[Thm.~4.3, p.~97]{Garey-Johnson-1979} with extra care to ensure parsimony.

We are given a \threedm{} instance with elements partitioned into sets $$X=\{x_1,\dots,x_n\},Y = \{y_1, \dots, y_n\},Z = \{z_1, \dots, z_n\}$$ and a set of triples $T \subseteq X \times Y \times Z$.
Let $m_T(x_i)$ be the multiplicity of $x_i$ in $T$, that is, the number of triples of $T$ where $x_i$ is the first coordinate, and similarly define $m_T(y_j)$ and $m_T(z_k)$.

First we pick a large base $B = 100 n$.
We use the notation $(d_5,d_4,d_3,d_2,d_1,d_0)_B$ to represent the base-$B$ number equal to $\sum_{i=0}^5 d_iB^i = d_5B^5 + d_4B^4 + d_3B^3 + d_2B^2 + d_1B + d_0$. In the discussion that follows, we use that $B$ is large enough that addition with a digit of the base-$B$ representation of the numbers in question will never carry over to another digit, so $(d_5,d_4,d_3,d_2,d_1,d_0)_B + (d_5',d_4',d_3',d_2',d_1',d_0')_B = (d_5+d_5', d_4+d_4',d_3+d_3',d_2+d_2',d_1+d_1',d_0+d_0')_B$.

We construct a \ndm{4} instance with target $t = (40,0,0,0,0,0)_B$ and multisets $W',X',Y',Z'$ having the following elements (also refer to Table~\ref{tab:n4dm}):

\begin{enumerate}[(a)]
\item For each $x_i\in X$, place $(10,i,0,0,0,0)_B$ in $W'$, $(11,0,-i,0,0,0)_B$ in $Z'$, and $m_T(x_i)-1$ copies of $(10,i,-i,0,0,0)_B$ in $W'$.
\item For each $y_j\in Y$, place $(10,0,0,j,0,0)_B$ in $X'$, $(12,0,0,0,-j,0)_B$ in $W'$, and $m_T(y_j)-1$ copies of $(10,0,0,j,-j,0)_B$ in $X'$.
\item For each $z_k\in Z$, place $(10,0,0,0,0,k)_B$ in $Y'$ and $(7,0,0,0,0,-k)_B$ in $X'$.
\item For each $(x_i,y_j,z_k)\in T$, place $(10,-i,0,-j,0,-k)_B$ in $Z'$ and $(10,0,i,0,j,k)_B$ in $Y'$.
\end{enumerate}

\begin{table}
\centering
\begin{tabular}{|c|c|c|c|c|}
\hline
& $x_i \in X$ & $y_j \in Y$ & $z_k \in Z$ & $(x_i,y_j,z_k) \in T$\\\hline
\multirow{2}{*}{$W'$} & $(10,i,0,0,0,0)_B$ & \multirow{2}{*}{$(12,0,0,0,-j,0)_B$}  &                                   & \\
& \boldmath $(10,i,-i,0,0,0)_B$    &                                       &                                   & \\\hline
\multirow{2}{*}{$X'$} &                    & $(10,0,0,j,0,0)_B$                    & \multirow{2}{*}{$(7,0,0,0,0,-k)_B$}& \\
&                                          & \boldmath $(10,0,0,j,-j,0)_B$ &                                   & \\\hline
$Y'$ &                                     &                                       & $(10,0,0,0,0,k)_B$                & $(10,0,i,0,j,k)_B$\\\hline
$Z'$ & $(11,0,-i,0,0,0)_B$                 &                                       &                                   & $(10,-i,0,-j,0,-k)_B$\\\hline
\end{tabular}
\caption{The constructions of $W',X',Y',Z'$. Each column represents the source of the constructed elements from the original \threedm{} instance.
Most elements have multiplicity~$1$; bold elements have multiplicity
one fewer than the corresponding \threedm{} source item.}
\label{tab:n4dm}
\end{table}

Importantly, the $x_i, y_j, z_k$ above are indexed starting at~$1$, not~$0$.
One can verify that the only quadruples summing to $t$ are the following
(given in $W',X',Y',Z'$ order):

\medskip

\newcounter{type}

\hbox to \textwidth{%
\arraycolsep=0pt
\begin{minipage}[t]{0.3\textwidth}
\refstepcounter{type}
\textbf{Type \thetype.}
For $(x_i,y_j,z_k) \in T$:
\label{form:pos}
$$
\begin{array}{llr@{\,}r@{\,}r@{\,}r@{\,}r@{\,}rl}
      &(&10,& i,&0,& 0,&0,& 0)_B\\
+\quad&(&10,& 0,&0,& j,&0,& 0)_B\\
+\quad&(&10,& 0,&0,& 0,&0,& k)_B\\
+\quad&(&10,&-i,&0,&-j,&0,&-k)_B\\
=\quad&(&40,& 0,&0,& 0,&0,& 0)_B
\end{array}
$$
\end{minipage}
\hss
\begin{minipage}[t]{0.3\textwidth}
\refstepcounter{type}
\textbf{Type \thetype.}
For $(x_i,y_j,z_k) \in T$:
\label{form:neg}
$$
\begin{array}{llr@{\,}r@{\,}r@{\,}r@{\,}r@{\,}rl}
      &(&12,&0,& 0,&0,&-j,& 0)_B\\
+\quad&(& 7,&0,& 0,&0,& 0,&-k)_B\\
+\quad&(&10,&0,& i,&0,& j,& k)_B\\
+\quad&(&11,&0,&-i,&0,& 0,& 0)_B\\
=\quad&(&40,&0,& 0,&0,& 0,& 0)_B\\
\end{array}
$$
\end{minipage}
\hss
\begin{minipage}[t]{0.35\textwidth}
\refstepcounter{type}
\textbf{Type \thetype.}
For $(x_i,y_j,z_k) \in T$:
\label{form:cleanup}
$$
\begin{array}{llr@{\,}r@{\,}r@{\,}r@{\,}r@{\,}rl}
      &(&10,& i,&-i,& 0,& 0,& 0)_B\\
+\quad&(&10,& 0,& 0,& j,&-j,& 0)_B\\
+\quad&(&10,& 0,& i,& 0,& j,& k)_B\\
+\quad&(&10,&-i,& 0,&-j,& 0,&-k)_B\\
=\quad&(&40,& 0,& 0,& 0,& 0,& 0)_B
\end{array}
$$
\end{minipage}%
}

\bigskip

Furthermore, there is a one-to-one correspondence between solutions to the source instance and solutions to the constructed instance by choosing a triple to be in the solution $S \subseteq T$ of the \threedm{} instance if and only if both the corresponding triples of the types~\ref{form:pos} and~\ref{form:neg} above are included in the solution $S'$ of the constructed \ndm{4} instance.
If a type-\ref{form:pos} quadruple is included in the solution for some $(x_i,y_j,z_k) \in T$, then the corresponding type-\ref{form:neg} quadruple must also be included because there is no other way to cover the element $(10,0,i,0,j,k)_B \in Y'$, and similarly for the reverse. 
To confirm that the rest of the elements can be covered by the type-\ref{form:cleanup} quadruples, notice that there are exactly the correct number of $(10,i,-i,0,0,0)_B \in W'$ and $(10,0,0,j,-j,0) \in X'$ elements.
In particular, for every $i$, there are $m_T(x_i)$ triples of the form $(10,-i,0,\ast,0,\ast)_B \in Z'$ and exactly one of these is covered by a type-\ref{form:pos} quadruple, so the remaining ones can be matched with the $m_T(x_i)-1$ elements of the form $(10,i,-i,0,0,0)_B \in W'$, and similarly for the $y_j$.  
Thus we have a parsimonious reduction from \threedm{} to \ndm{4}.

We can verify the claim that $Y' \cup (Y'+Z')$ is a set (not a multiset) using the initial assumption that no two triples in $T$ agree on more than one coordinate.
First, $Y'$ is a set because, for each $z_k$, there is exactly one element $(10,0,0,0,0,k)_B \in Y'$, and for each triple $(x_i,y_j,z_k) \in T$, there is exactly one element $(10,0,i,0,j,k)_B \in Y'$; these two types of elements are disjoint because the third and fifth digits are always zero in the former but nonzero in the latter.
Similarly, $Z'$ is a set (a fact we will need later).
Also, $Y'$ and $Y'+Z'$ are disjoint because the first digit of any element of $Y'$ is $10$ while the first digit of any element of $Y'+Z'$ is at least $20$.

To see that $Y'+Z'$ is a set, consider two equal sums $s_1 = y_1' + z_1'$ and
$s_2 = y_2' + z_2'$ for $y_1',y_2' \in Y'$ and $z_1',z_2' \in Z'$. From $s_1 =
s_2$, it follows that $y_2'-y_1' = z_2'-z_1'$. We claim that, if $s_1 =
s_2$, then $z_1' = z_2'$ and thus $y_1' = y_2'$, which suffices because we
argued that $Y'$ and $Z'$ are sets.
To prove the claim, we have two cases, one for each type of element of~$Z'$:

\paragraph{Case 1:} If $z_1' = (11,0,-i_1,0,0,0)_B$, then $z_2' = (11,0,-i_2,0,0,0)$ or else $s_1$ and $s_2$ would differ in the first digit. Thus $y_2'-y_1' = z_2'-z_1' = (0,0,i_1-i_2,0,0,0)_B$, but by the assumption that there are no two distinct triples of $T$ sharing both $y_j$ and $z_k$, there are no two distinct elements of $Y'$ whose last two digits are equal but whose third digits are not, so this difference is impossible unless $y_1' = y_2'$ and $z_1' = z_2'$.

\paragraph{Case 2:} If $z_1' = (10,-i_1,0,-j_1,0,-k_1)_B$, then $z_2' = (10,-i_2,0,-j_2,0,-k_2)_B$ or else $s_1$ and $s_2$ would differ in the first digit. Thus $y_2' - y_1' = z_2'-z_1' = (0,i_1-i_2,0,j_1-j_2,0,k_1-k_2)_B$, but no elements of $Y'$ have nonzero second or fourth digits, so it must be that $i_1 = i_2$ and $j_1 = j_2$. By the assumption that no distinct triples of $T$ share both $x_i$ and $y_j$, it must be that $k_1 = k_2$ as well, so $z_1' = z_2'$ as claimed.
\end{proof}

\begin{theorem} \label{thm:n3dm}
  \ndm{3} is strongly ASP-hard and \#P-hard,
  even if $X$ is required to be a set (not a multiset).
\end{theorem}

\begin{proof}
We give a parsimonious reduction from \ndm{4} where $W \cup (W+X)$ is a set,
as guaranteed by Theorem~\ref{thm:n4dm} (relabelling $Y,Z$ to $W,X$). 
Our reduction is essentially Garey and Johnson's reduction from
4-\textsc{Partition} to 3-\textsc{Partition}
\cite[Thm.~4.4, p.~99]{Garey-Johnson-1979}, with some extra care regarding
identical elements and splitting elements into separate sets $X',Y',Z'$.

Following the reduction in \cite{Garey-Johnson-1979}, given a \ndm{4}
instance $W = \{w_1,\dots,w_n\}$, $X = \{x_1,\dots,x_n\}$,
$Y=\{y_1,\dots,y_n\}$, $Z=\{z_1,\dots,z_n\}$ with target~$t$,
we will construct a \ndm{3} instance $X',Y',Z'$
with target sum $t' = (64,4)_B$ in base $B = t$.
We assume without loss of generality that every element of $W,X,Y,Z$ is strictly
between $t/5$ and $t/3$.%
\footnote{If a \ndm{4} instance has any elements $\geq t$, it trivially has no solutions (as all elements are positive).  Otherwise, we can convert it to an instance with this property by adding $2t$ to each element in $W,X,Y,Z$ and changing the target sum from $t$ to $\hat t = 9 t$. Then every element is strictly between $2t$ and $3t$, and thus strictly between $\hat t/5 = 9t / 5$ and $\hat t/3 = 9t / 3$.}
First we define the elements that will appear in
$X' \cup Y' \cup Z'$:
\begin{align*}
w'_i &= (21, 4 w_i+1)_B, \\
x'_j &= (19, 4 x_j+1)_B, \\
y'_k &= (19, 4 y_k+1)_B, \\
z'_\ell &= (21, 4 z_\ell+1)_B, \\
u[w_i, x_j] &= (24, -4(w_i+x_j)+2)_B, \\
\bar{u}[w_i,x_j] &= (20, \phantom{-}4(w_i+x_j)+2)_B, \\
C &= (20, 0)_B.
\end{align*}
Now we can construct the desired \ndm{3} instance, splitting these
elements into three multisets $X',Y',Z'$:
\begin{align*}
X'&=\{w'_i : 1 \leq i \leq n\} \cup \{\bar{u}[w_i,x_j] : 1 \leq i,j \leq n\},\\
Y'&=\{x'_j : 1 \leq j \leq n\} \cup \{z'_\ell : 1 \leq \ell \leq n\} \cup \{ n^2 - n \text{ copies of } C\}, \\
Z'&=\{u[w_i,x_j] : 1 \leq i,j \leq n\} \cup \{y'_k : 1 \leq k \leq n\}.
\end{align*}



There are $2 \times 3 \times 2 = 12$ possible forms of triples, shown below grouped by the equivalence classes modulo $4$ of the second coordinate of their sum (with shaded boxes to indicate the only triples that will turn out to be valid):

\begin{center}
\def\,{,~}
\definecolor{valid}{rgb}{0.85,1,0.85}
\begin{tabular}{cccccc}
\boldmath $0 \pmod 4$ & \boldmath $1 \pmod 4$ & \boldmath $2 \pmod 4$ & \boldmath $3 \pmod 4$ \\ \hline
\multicolumn{1}{|c|}{\cellcolor{valid}
$(w'_{i'}\, x'_{j'}\, u[w_i, x_j])$}          & $(\bar{u}[w_{\bar\imath},x_{\bar\jmath}]\, x'_{j'}\, u[w_i, x_j])$    & $(w'_{i'}\, C\, y'_{k'})$ & $(w'_{i'}\, x'_{j'}\, y'_{k'})$ \\
\hhline{|-}
$(w'_{i'}\, z'_{\ell'}\, u[w_i, x_j])$       & $(\bar{u}[w_{\bar\imath},x_{\bar\jmath}]\, z'_{\ell'}\, u[w_i, x_j])$ &                   & $(w'_{i'}\, z'_{\ell'}\, y'_{k'})$ \\
$(\bar{u}[w_{\bar\imath},x_{\bar\jmath}]\, x'_{j'}\, y'_{k'})$     &                                            &                   & $(w'_{i'}\, C\, u[w_i, x_j])$ \\
\hhline{|-}
\multicolumn{1}{|c|}{\cellcolor{valid}
$(\bar{u}[w_{\bar\imath},x_{\bar\jmath}]\, z'_{\ell'}\, y'_{k'})$} &                                            &                   & $(\bar{u}[w_{\bar\imath},x_{\bar\jmath}]\, C\, y'_{k'})$ \\
\hhline{|-}
\multicolumn{1}{|c|}{\cellcolor{valid}
$(\bar{u}[w_{\bar\imath},x_{\bar\jmath}]\, C\, u[w_i, x_j])$} &                                            &                   & \\
\hhline{|-}
\end{tabular}
\end{center}

The second coordinate of $t'$ is congruent to $0 \pmod 4$, so triples in the second, third, and fourth columns never sum to $t'$.

Of the triples in the first column, two of them cannot actually sum to~$t'$.
Triples of the form $(w'_{i'}, z'_{\ell'}, u[w_i, x_j])$ sum to $(66, 4(w_{i'} + z_{\ell'} - w_i - x_j)+4)_B$.  For this to equal $t'$, the second coordinate must equal $-2t+4$, so $w_{i'} + z_{\ell'} - w_i - x_j$ must equal $-t/2$.  But by the assumption that every element of $W,X,Y,Z$ is strictly between $t/5$ and $t/3$, the smallest possible value for $w_{i'} + z_{\ell'} - w_i - x_j$ is greater than $-4t/15$, so triples of this form never sum to $t'$.  Similarly, triples of the form $(\bar{u}[w_{\bar\imath},x_{\bar\jmath}], x'_{j'}, y'_{k'})$ sum to $(58, 4(w_{\bar\imath} + x_{\bar\jmath} + x_{j'} + y_{k'})+4)_B$.  For this to equal $t'$, the second coordinate must equal $6t+4$, so $w_{\bar\imath} + x_{\bar\jmath} + x_{j'} + y_{k'}$ must sum to $3t/2$, but the largest possible value for that expression is less than $4t/3$, so triples of this form never sum to $t'$.

This leaves three forms of triples that can sum to $t'= (64,4)_B$:
\[
(w'_{i'}, x'_{j'}, u[w_i, x_j]), \quad (\bar{u}[w_{\bar\imath},x_{\bar\jmath}], z'_{\ell'}, y'_{k'}), \quad \text{and} \quad (\bar{u}[w_{\bar\imath},x_{\bar\jmath}], C, u[w_i,x_j])\text{.}
\]
A triple of the second form encodes a quadruple in the input \ndm{4} instance; the triple sums to $t'$ (after a carry $(60,4t+4)_B = (64,4)_B$) exactly when $w_{\bar\imath} + x_{\bar\jmath} + z_{\ell'} + y_{k'} = t$.  The map $w_{\bar\imath} + x_{\bar\jmath} \mapsto \bar{u}[w_{\bar\imath},x_{\bar\jmath}]$ is one-to-one, so from our assumption that $W+X$ is a set, $\{\bar{u}[w_i,x_j]\}$ is also a set, and so this encoding is unique.  A triple of the first form sums to $t'$ exactly when $w_{i'}+x_{j'}-w_i-x_j = 0$. Because $W+X$ is a set and the map $w_i+x_j \mapsto u[w_i,x_j]$ is one-to-one, we must have $i' = i$ and $j' = j$ in valid triples of the first form, uniquely collecting the $u[w_i, x_j]$ elements corresponding to $\bar{u}[w_{\bar\imath},x_{\bar\jmath}]$ elements used in triples of the second form.  Similarly, $\bar\imath=i$ and $\bar\jmath=j$ in valid triples of the third form, uniquely collecting the unused $u[w_i,x_j]$ and $\bar{u}[w_i,x_j]$ elements using all $n^2-n$ copies of $C$.  Thus there is a one-to-one correspondence between solutions to the input \ndm{4} instance and the constructed \ndm{3} instance, so the reduction is parsimonious and \ndm{3} is ASP- and \#P-hard.

It remains to verify that $X' = \{w'_i\} \cup \{\bar{u}[w_i,x_j]\}$ is a set (not a multiset).  We argued above that $\{\bar{u}[w_i,x_j]\}$ is a set (using that $W+X$ is a set), and $\{w'_i\}$ is a set because we assumed $W$ is a set and the map $w_i \mapsto w'_i = (20,4 w_i+1)_B$ is one-to-one.
It remains to show that $\{w'_i\}$ is disjoint from $\{\bar{u}[w_i,x_j]\}$,
which follows
because $w'_i \equiv 1 \pmod 4$ and $\bar{u}[w_i,x_j] \equiv 2 \pmod 4$.
Therefore $X'$ is a set.
\end{proof}

\section{Parsimonious Reductions from Numerical 3DM to Path Puzzles}
\label{numerical3dm-to-path-puzzles}
The goal of this section is to parsimoniously reduce \ndm3
(as analyzed in Section~\ref{ndm})
to \pathpuzzle, thereby proving the latter strongly NP-, ASP-, and \#P-hard.
We first introduce a more geometric view of \ndm3, called
\lengthoffsets, and prove its equivalence.
It will then be relatively easy to represent \lengthoffsets\
as a \pathpuzzle.

\begin{problem}[\problemfont\lengthoffsets]
Given a set (not a multiset) of positive integer lengths
$a_1, a_2, \allowbreak \dots, a_n$,
and given $m$ nonnegative integer target densities
$t_0, t_1, \dots, t_{m-1}$,
can we place $n$ intervals with integer endpoints within $[0,m]$
and lengths $a_1, a_2, \dots, a_n$, respectively,
such that the number of intervals overlapping $(i,i+1)$ is exactly
the target density~$t_i$?
In other words, can we choose nonnegative integer offsets
$b_1, b_2, \dots, b_n$ such that
$a_j + b_j \leq m$ for each $j$ ($1 \leq j \leq n$);
and, for each $i$ ($0 \leq i < m$),
there are exactly $t_i$ indices $j$ such that $b_j \le i < a_j + b_j$?
\end{problem}


\begin{theorem} \label{thm:length-offsets}
\lengthoffsets{} is parsimoniously reducible from \ndm{3} in which at least one of the three multisets is actually a set.
\end{theorem}

\begin{proof}
We give a parsimonious reduction from \ndm{3} where $X$ is a set,
as guaranteed by Theorem~\ref{thm:n3dm}.
Specifically, consider a \ndm{3} instance with set $X = \{x_1,
\ldots, x_n\}$, multisets $Y$ and $Z$, and a target sum $t$.
Assume without loss of generality that every element of $X,Y,Z$ is strictly
between $t/4$ and $t/2$.%
\footnote{If a \ndm{3} instance has any elements $\geq t$, it trivially has no solutions (as all elements are positive).  Otherwise, we can convert it to an instance with this property by adding $t$ to each element in $X,Y,Z$ and changing the target sum from $t$ to $\hat t = 4 t$. Then every element is strictly between $t$ and $2t$, and thus strictly between $\hat t/4 = 4t / 4$ and $\hat t/2 = 4t / 2$.}
We construct a \lengthoffsets{} instance that we claim has
the same number of solutions: the $n$ lengths are given simply by $a_i = x_i$,
and the target densities are given by
$t_i = n - |\{y \in Y : y > i\}| - |\{z \in Z : t-z \leq i\}|$,
where $0 \leq i < m$ and $m=t$.
See Figure~\ref{fig:lenoff} for an example.
The intuition is that we place intervals for $X,Y,Z$,
left-align the intervals for $Y$, right-align the intervals for~$Z$,
and count the remaining density for $X$ intervals.

\begin{figure}
\centering
\includegraphics[scale=0.75]{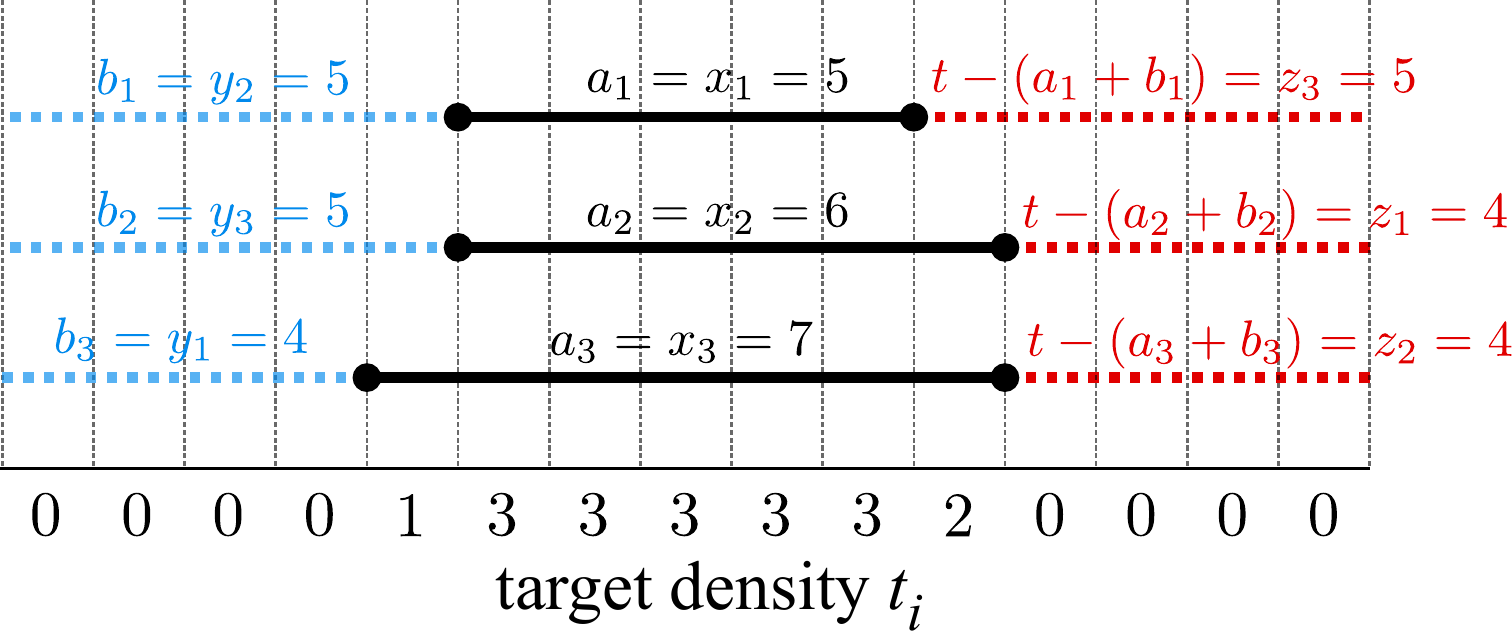}
\caption{\lengthoffsets{} instance obtained by reducing from \ndm{3}
  where $X = \{5,6,7\}, Y = \{4,5,5\}, Z = \{4,4,5\}, t = 15$; and
  its solution corresponding to the \ndm{3} solution of
  $(5,5,5)$, $(5,6,4)$, $(4,7,4)$.}
\label{fig:lenoff}
\end{figure}

It remains to show that every solution to the original \ndm{3} instance corresponds to a solution to the constructed \lengthoffsets{} instance, and different \ndm{3} solutions correspond to different \lengthoffsets{} solutions. Equivalently, we will provide an injective map from \ndm{3} solutions to \lengthoffsets{} solutions, and an injective map from \lengthoffsets{} solutions to \ndm{3} solutions.

\paragraph{N3DM to \lengthoffsets.} To convert a \ndm{3} solution into a \lengthoffsets{} solution, we assign $b_j = y_k$ for each solution triple $(x_j, y_k, z_\ell)$.  Figure~\ref{fig:lenoff} shows this solution for the example.

For each $i$ and each triple $(x_j,y_k,z_\ell)$,
$b_j \le i < a_j + b_j$ if and only if
$y_k \le i$ and $i < x_j + y_k = t - z_\ell$, so either
\be
\ii $y_k > i$; or 
\ii $t - z_\ell \le i$; or 
\ii $b_j \le i < a_j + b_j$. 
\ee
The first case applies $|\{y \in Y : y > i\}|$ times, and
the second case applies $|\{z \in Z : t-z \le i\}|$ times,
so the third case applies
$n - |\{y \in Y : y > i\}| - |\{z \in Z : t-z \le i\}| = t_i$ times.
Thus our choice of the offsets is a valid solution to the
\lengthoffsets{} instance.

If two \ndm{3} solutions differ, then (using that $X$ is a set) some $x$ is matched with a different $y$ in each solution, so when converting those solutions to \lengthoffsets{} solutions, we assign the corresponding length different offsets.

\paragraph{\lengthoffsets{} to N3DM.} To convert a \lengthoffsets{} solution into a \ndm{3} solution, for each length--offset pair $(a_i, b_i)$, we match the triple $(a_i, b_i, t-a_i-b_i)$.
These triples obviously sum to $t$ and are therefore legal, but we need to
show that their elements exist and cover $X$, $Y$, and $Z$, respectively.

\be
\ii Every $a_i$ is an $x_i$ and vice versa, so $X$ is covered by $\{a_i\}$. 
\ii For each $i$, we have
\begin{align*}
 t_{i} - t_{i-1} &= (n - |\{y \in Y: y > i\}| - |\{z \in Z: t-z \le i\}|) \\
&\phantom{=}\, - (n - |\{y \in Y: y > i-1\}| - |\{z \in Z: t-z \le i-1\}|)\\
 &= |\{y \in Y: y = i\}| - |\{z \in Z: t-z = i\}|
\end{align*}

For $i < t/2$, the second term is $0$ (because $z < t/2$ by assumption), so
$t_i - t_{i-1}$ is precisely the number of elements of $Y$ that equal $i$. On
the other hand, in a \lengthoffsets{} instance, when $i < t/2$, $t_i -
t_{i-1}$ is the number of segments which pass through $i$ but not $i-1$, i.e.,
the number of segments which begin at $i$ and therefore have offset $b_j = i$.
Therefore, $Y$ is covered by $\{b_j\}$.

\ii Following the same argument as above, but for $i > t/2$, we have that $t_i - t_{i-1} = -|\{z \in Z: t-z = i\}|$. 
On the other hand, in the \lengthoffsets{} problem, for $i > t/2$, $t_i - t_{i-1}$ is the negative of the number of 
segments which \emph{end} at $i-1$; i.e., it is the negative of the number of indices $j$ such that
$a_j + b_j = i$, or equivalently, $t - (t - (a_j+b_j)) = i$. Thus, $Z$ is covered by $\{t-(a_j+b_j)\}$ as claimed. 
\ee

Different solutions to the \lengthoffsets{} instance correspond to different
solutions to the \ndm{3} instance because, if two \lengthoffsets{} solutions
differ, then some length $a_j$ gets different offsets, so the corresponding
$x_j$ is matched to different elements of $Y$ in the two \ndm{3} solutions.

We have shown two injective maps between solutions of the \lengthoffsets{}
instance and solutions of the \ndm{3} instance,
so our reduction is parsimonious.
\end{proof}

We now make a brief observation that we make use of later.

\begin{lemma}
\label{thm:endpoint-sets}
In every solution to \lengthoffsets{} instances produced by the reduction in the proof of Theorem~\ref{thm:length-offsets}, no line segment shares its left endpoint with the right endpoint of another.
\end{lemma}

\begin{proof}
By assumption, all elements of the \ndm{3} instance lie in the exclusive interval $(t/4,t/2)$.  Our reduction sets $a_i = x_i$, so $t/4 < a_i < t/2$.  Our mapping between solutions assigns $b_i = y_k$, so $t/4 < b_i < t/2$.  Adding these inequalities yields $t/2 < a_i + b_i < t$ for all $i$.  Then $b_i < t/2 < a_j + b_j$ for all $i$ and $j$, so the sets of left and right endpoints are disjoint.
\end{proof}

We are now ready to prove our main theorem.

\begin{figure}
\centering
\includegraphics[width=\textwidth]{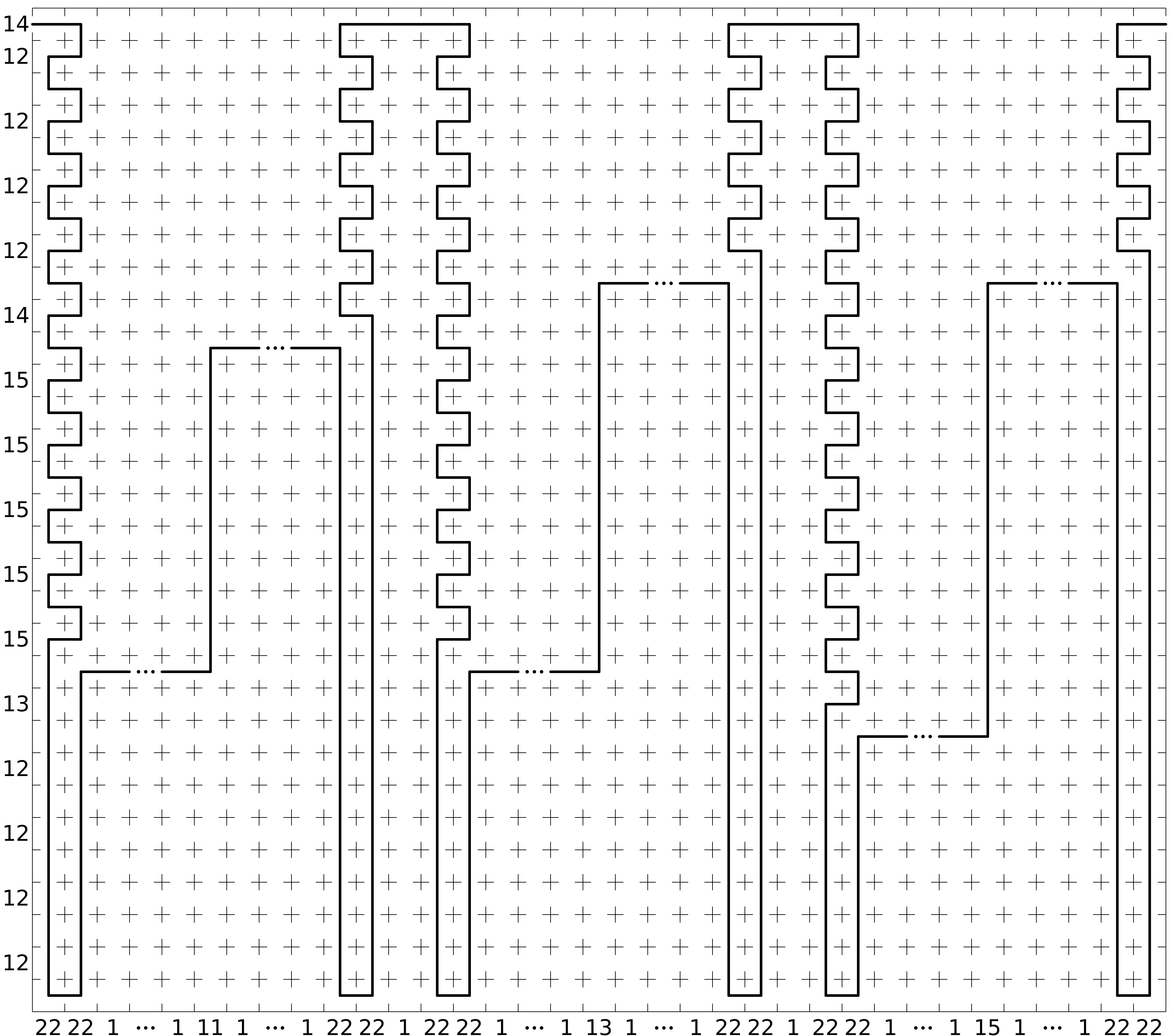}
\caption{Solution to a \pathpuzzle{} instance, reduced from \lengthoffsets{} from Figure~\ref{fig:lenoff} where $n = 3, m = 15, a_i = (5,6,7), t_i = (0,0,0,0,1,3,3,3,3,3,2,0,0,0,0)$.
Ellipses elide sections of $6n=18$ columns each labeled~$1$. }
\label{fig:path-puzzle}
\end{figure}

\begin{theorem}
\pathpuzzle{} is NP-, \#P- and ASP-hard.
\label{thm:lo-to-pp}
\end{theorem}

\begin{proof}
We give a parsimonious reduction from \lengthoffsets, as produced by
Theorem~\ref{thm:length-offsets} so that Lemma~\ref{thm:endpoint-sets}
applies.  Given a \lengthoffsets\ problem with lengths $a_1, \dots, a_n$ and
target densities $t_0, \dots, t_{m-1}$, we construct an equivalent
\pathpuzzle\ instance as follows.
Figure~\ref{fig:path-puzzle} shows our construction instantiated for the same
\lengthoffsets{} instance from Figure~\ref{fig:lenoff}.

\begin{description}
\item[Dimensions:]
  The grid has $2m+3$ rows and $(12n+6)n-1$ columns.

  We group the columns into $n$ \emph{blocks} $B_1,\dots,B_n$
  of $12n+5$ columns each, interspersed with $n-1$ \emph{lone columns}.
  Thus block $B_i$ ($1 \leq i \leq n)$
  consists of columns $(12n+6)i-(12n+5),\dots,(12n+6)i-1$
  and the $i$th lone column ($1 \leq i < n$) is column $(12n+6)i$.
\item[Doors:] We place doors at the left and right ends of the top row.
\item[Row labels:] Counting up from the bottom,
  the $(2i+2)$nd row has a label of $4n + t_i$, for each $i$ ($0 \le i \le m$);
  the $(2m+2)$nd and $(2m+3)$rd (topmost) rows have labels $4n$ and $5n-1$,
  respectively; and all other row labels are blank.
\item[Column labels:]
  Each lone column has a column label of~$1$.
  Each block $B_j$ ($1 \leq j \leq n$) has labels of $2m+3$ on its
  first two and last two columns; a label of $2 a_j + 1$ on its middle
  column ($6n+3$th column); and labels of $1$ on all other columns
  (which split into two sections of $6n$ consecutive columns).
\end{description}


Any solution to this path puzzle has the following properties:
\be
\ii \label{filled columns}
    Every square in the first two and last two columns of each block is
    visited, by the column labels of $2m+3$.
\ii Every section of $6n$ consecutive columns labeled $1$ (within a block)
    corresponds to a single horizontal path, which must be in one of the blank
    rows because all row labels are less than~$6n$.
\ii \label{vertical lines}
    Every column labeled $2 a_i + 1$ is a single vertical line segment,
    because both neighboring columns are labeled $1$, just enough
    to enter and exit the column once.
\ii No top square of a column labeled $2 a_i + 1$ is visited, because if
    one were, the second square from the top of such a column would also be
    visited (by Property~\ref{vertical lines} and because $2 a_i + 1 > 1$),
    but then that row would have more than $4n$ visited squares
    (by Property~\ref{filled columns}).
\ii In each line column, the top square (and only that square) is visited,
    because those are the remaining squares on the top row that can be visited,
    and are just enough (with Property~\ref{filled columns})
    to account for a total of~$5n$.
\ii The vertical line segment in the (unique) column labeled $2a_i+1$
    (from Property~\ref{vertical lines}) visits $a_i$ rows with labels of
    $4n+t_j$ for various~$j$. Of that $4n+t_j$, $4n$ visits are
    accounted for by the full columns of Property~\ref{filled columns},
    so the positions of those line segments are a solution to the
    \lengthoffsets{} problem.
\ii Given placements of the vertical line segments corresponding to a
    valid solution to the \lengthoffsets{} problem, we claim that
    the rest of the path is uniquely determined.
    The set of visited squares in each gadget is uniquely determined
    by the previous properties.  In each pair of columns labeled $2m+3$:
\be
\ii The bottom two squares each have only two visited neighbors, each other
and the square above them, so each of them connects by the path to those
two squares.
\ii For squares in the pair of columns
below the entry point of the length $6n$ horizontal path, the long U-shaped 
path shown in Figure~\ref{fig:path-puzzle} is forced.  For each horizontal pair of squares except the bottom pair, the squares below connect to them.  If the pair squares connect to each other, they form a closed loop, so they must connect to the squares above them instead.
\ii For squares above the entry point of the length $6n$ 
horizontal path, the zig-zag path is forced.
The pair of columns divides evenly into $2 \times 2$ chunks 
because the entry point of the length $6n$ path is in a row with no label, and the only such rows are at even height.  Let \emph{inside} and \emph{outside} be relative to the center of the block.  In each chunk, the bottom
inside square can't connect to the square
below (because that square is already known to connect down and to the
inside), so it connects to the outside and up.  Similarly, the top outside square
can't connect to the square below (because that square is
already known to connect down and to the bottom inside square),
so it connects to the inside and up (except that in the very top $2
\times 2$ chunk, the top outside square can't connect down and can and must
connect to the outside to satisfy the top row).
\ee
\ee
Thus, each solution to the path puzzle determines a
solution to the \lengthoffsets{} problem, and that solution is uniquely
determined, so the number of solutions to the \lengthoffsets{} problem is
the same as the number of solutions of the path puzzle, and the
reduction is parsimonious as desired.
Note that we are relying on the uniqueness of the lengths $a_i$ from the \lengthoffsets{} problem definition; otherwise, permuting which copy of a duplicated length gets which offset in the path puzzle would generate multiple solutions to \pathpuzzle{} from each solution of \lengthoffsets{}.
\end{proof}

In fact, our reduction can be converted into one giving
complete information (i.e., all row and column labels), demonstrating that
partial information is not the source of \pathpuzzle{}'s hardness.

\begin{theorem}
Perfect-information \pathpuzzle{} (with all row and column labels given as labels) is NP-, \#P- and ASP-hard.
\end{theorem}

\begin{proof}
Recall that the reduction from the proof of Theorem~\ref{thm:lo-to-pp}
(referring to Figure~\ref{fig:path-puzzle})
already provides all column sums and about half of the row sums.
We show how to provide the remaining row sums without giving away information
about the solution to the original \lengthoffsets{} instance.

The rows with missing labels are the $(2i-1)$st rows for $i=1,2,\dots,m$.
In each such row, the solution path must visit $(6n+2)r_i$ cells,
where $r_i$ is the number of segments in the \lengthoffsets{} 
solution which have an endpoint at $i$.
Recall from Lemma~\ref{thm:endpoint-sets} that no line segment shares its
left endpoint with the right endpoint of another.  Thus there is only one type
of endpoint at each coordinate $i$, and we can compute $r_i = |t_{i+1}-t_i|$.
The value of $r_i$ depends only on the \lengthoffsets{} instance, not its
solutions, so we can modify our reduction to specify a label of $(6n+2)r_i$
for row $2i-1$, producing a perfect-information instance of \pathpuzzle{}.
\end{proof}

\section{Open Problems}

One interesting open problem is whether Planar 3DM, where the bipartite
graph of elements and triples is planar, is also ASP-hard and \#P-hard.
This problem is known to be NP-hard \cite{Dyer-Frieze-1986}, and
the variable--clause gadget structure in the proof of Theorem~\ref{thm:3dm}
is close to preserving planarity.  Unfortunately, the initial clause tripling
destroys any planarity in the input, and seems difficult to avoid.

Another intriguing open problem is whether discrete tomography with
partial information, but no Hamiltonian path constraint, is NP-hard.
If true, this would be another aspect of path puzzles which make them hard.

\section*{Acknowledgements}

We thank Jayson Lynch for useful discussions and debugging help, and Quanquan Liu for help in constructing the figures for this paper.
Most figures were produced using SVG Tiler (\url{https://github.com/edemaine/svgtiler}).

\bibliography{bibliography}
\bibliographystyle{alpha}

\appendix

\clearpage

\section{Solution to the Font Puzzles}
\label{app:font-sol}

\begin{figure}[h]
  \centering
  \includegraphics[width=.13\linewidth]{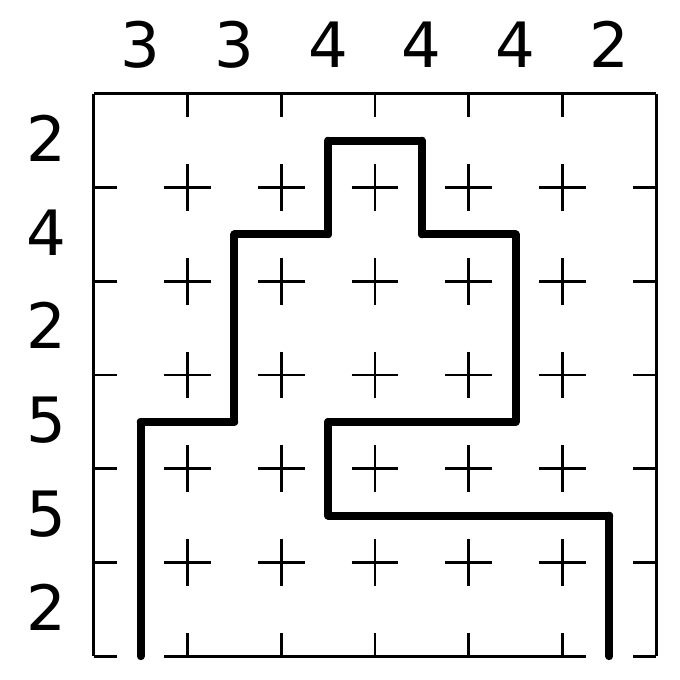}
  \includegraphics[width=.13\linewidth]{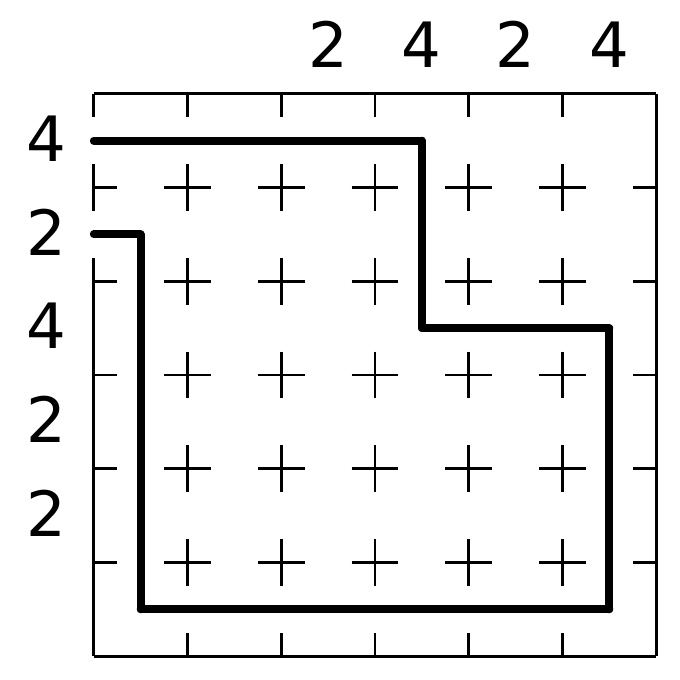}
  \includegraphics[width=.13\linewidth]{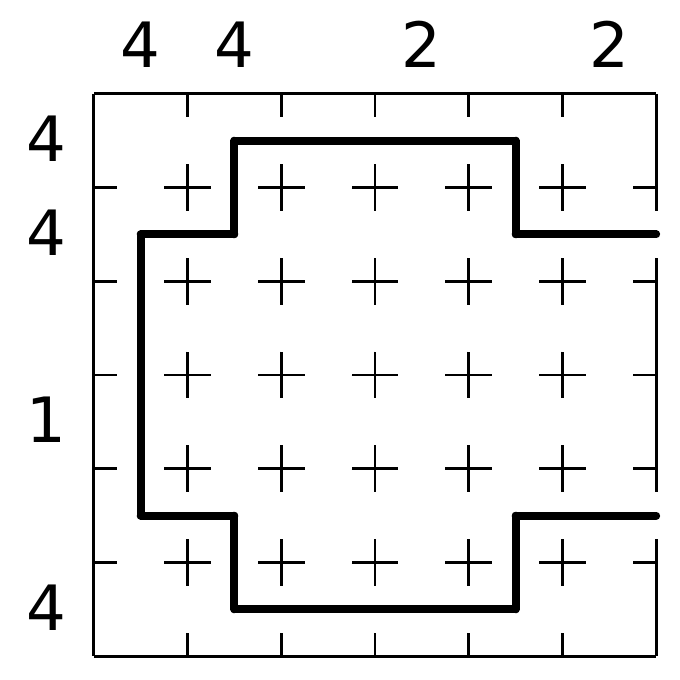}
  \includegraphics[width=.13\linewidth]{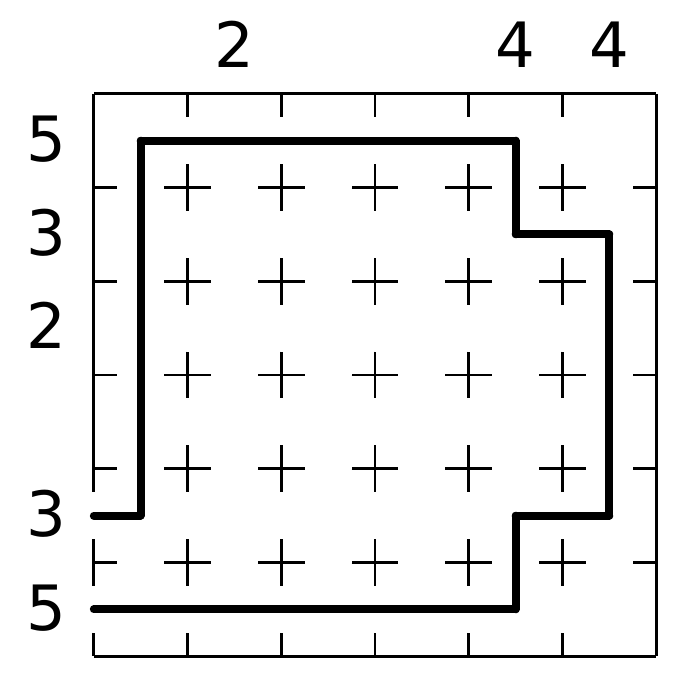}
  \includegraphics[width=.13\linewidth]{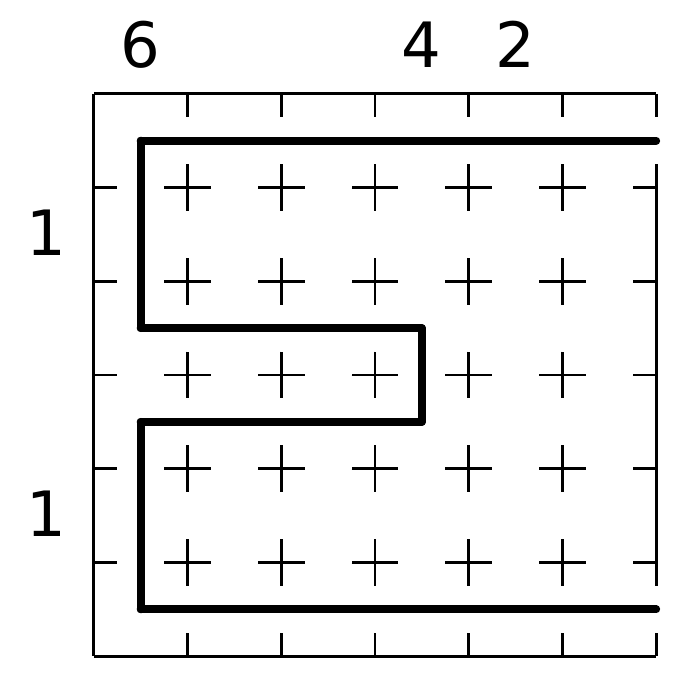}
  \includegraphics[width=.13\linewidth]{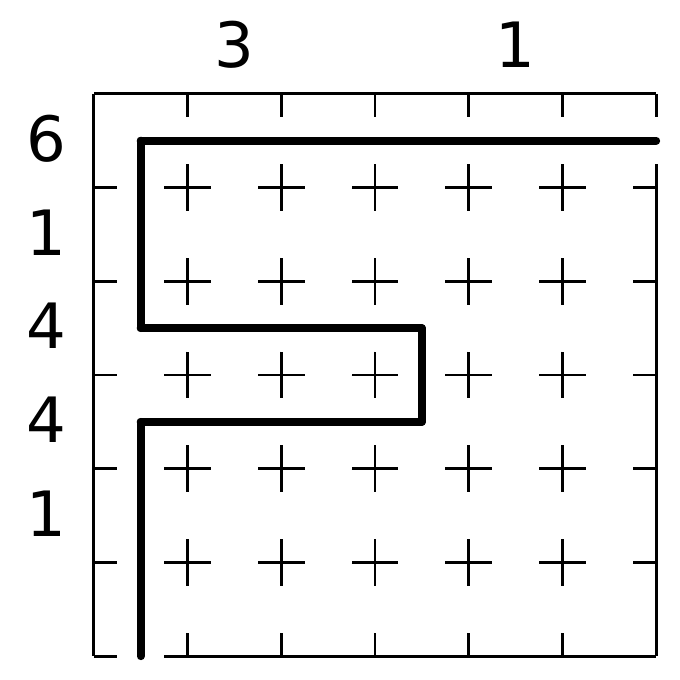}
  \includegraphics[width=.13\linewidth]{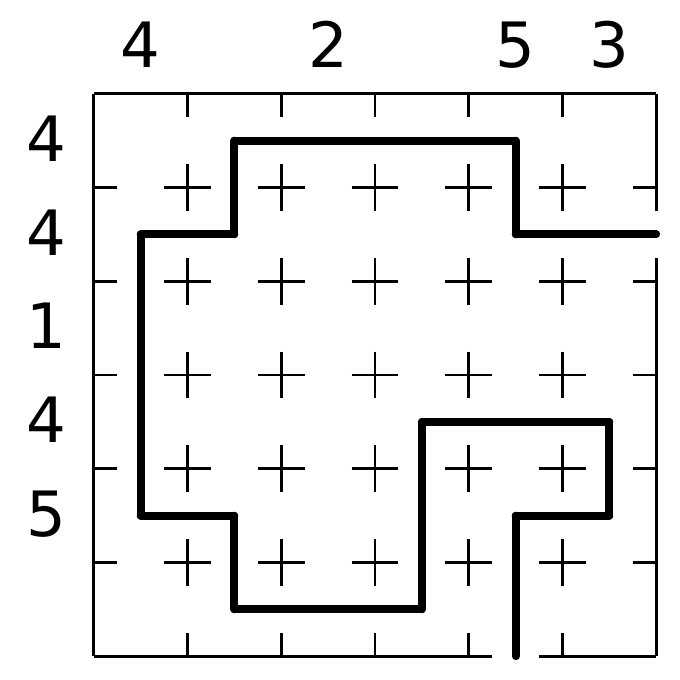}

  \includegraphics[width=.13\linewidth]{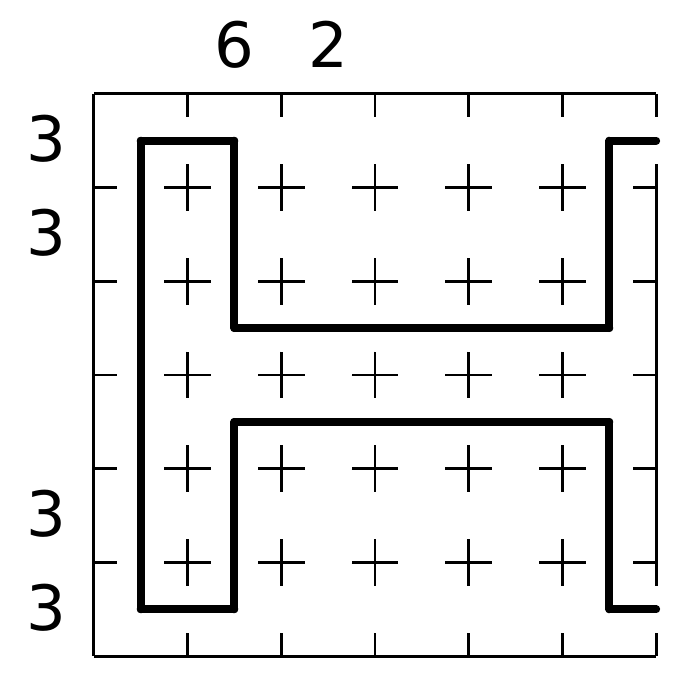}
  \includegraphics[width=.13\linewidth]{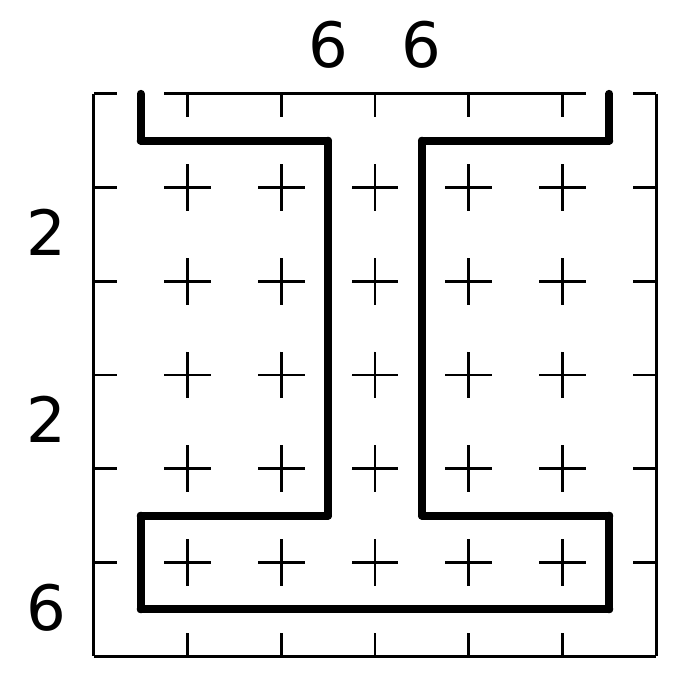}
  \includegraphics[width=.13\linewidth]{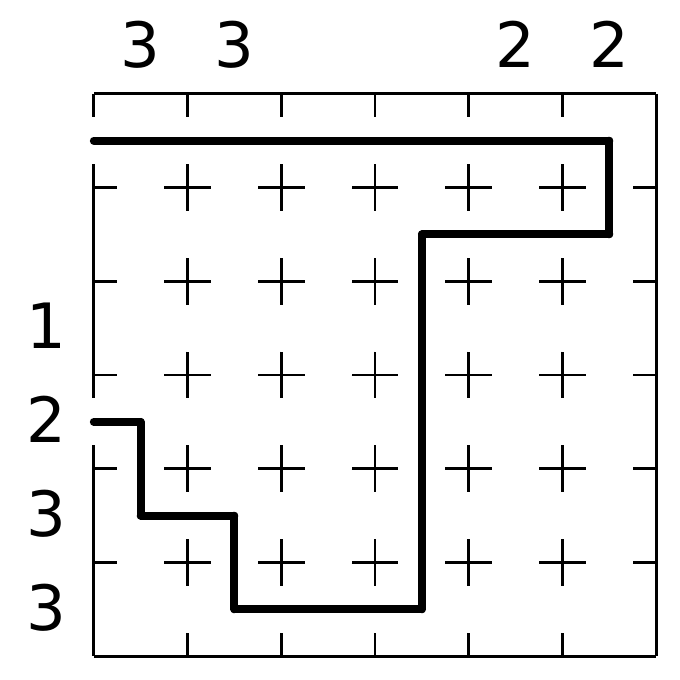}
  \includegraphics[width=.13\linewidth]{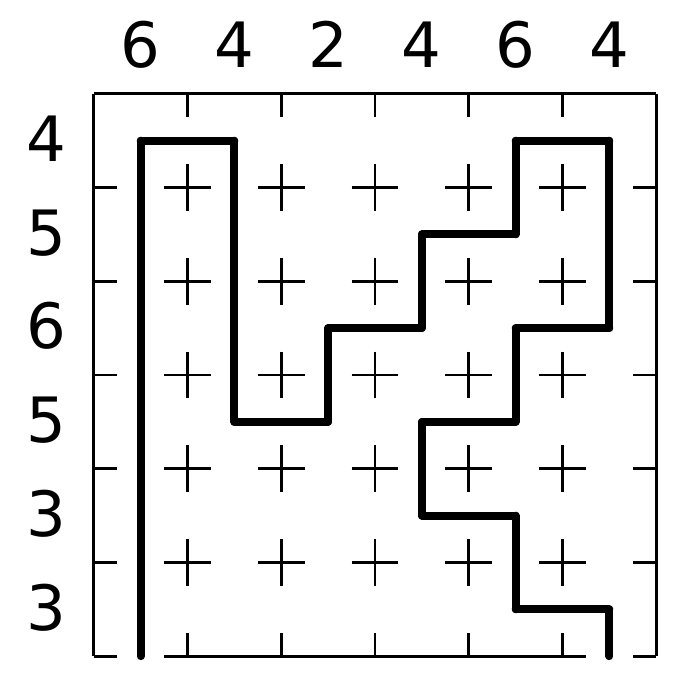}
  \includegraphics[width=.13\linewidth]{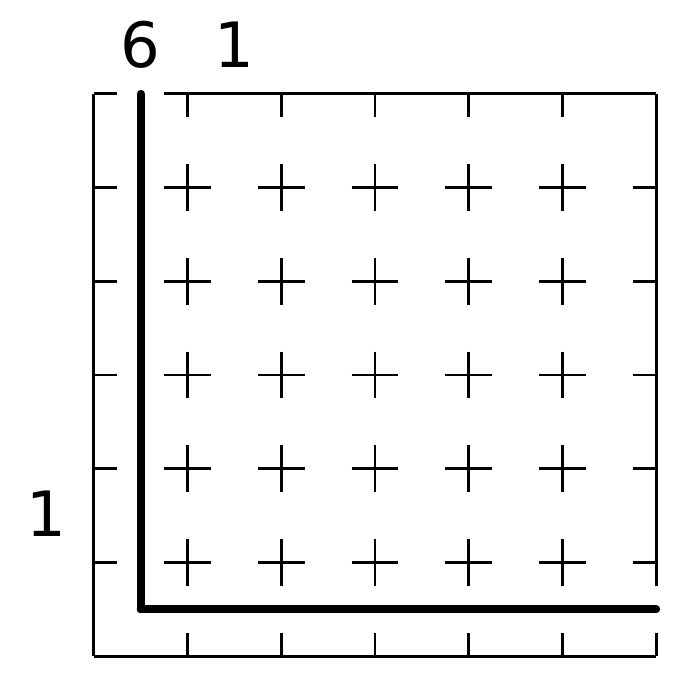}
  \includegraphics[width=.13\linewidth]{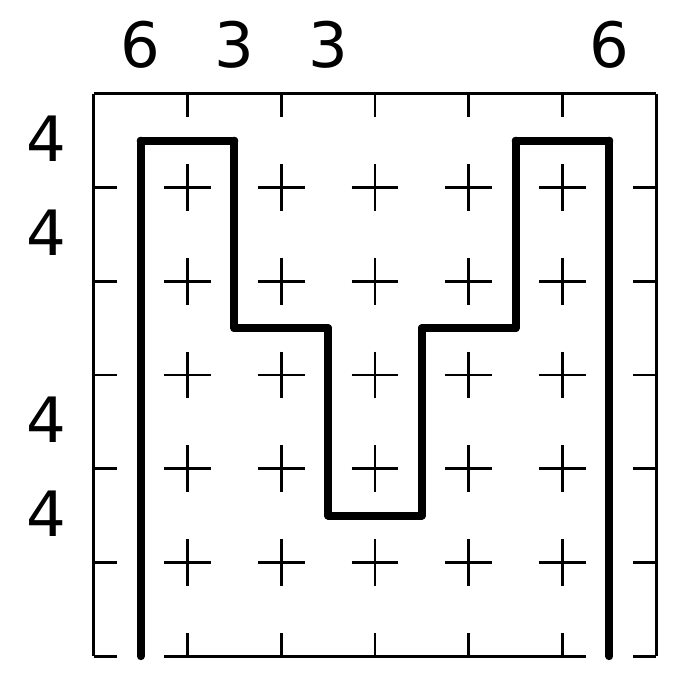}
  \includegraphics[width=.13\linewidth]{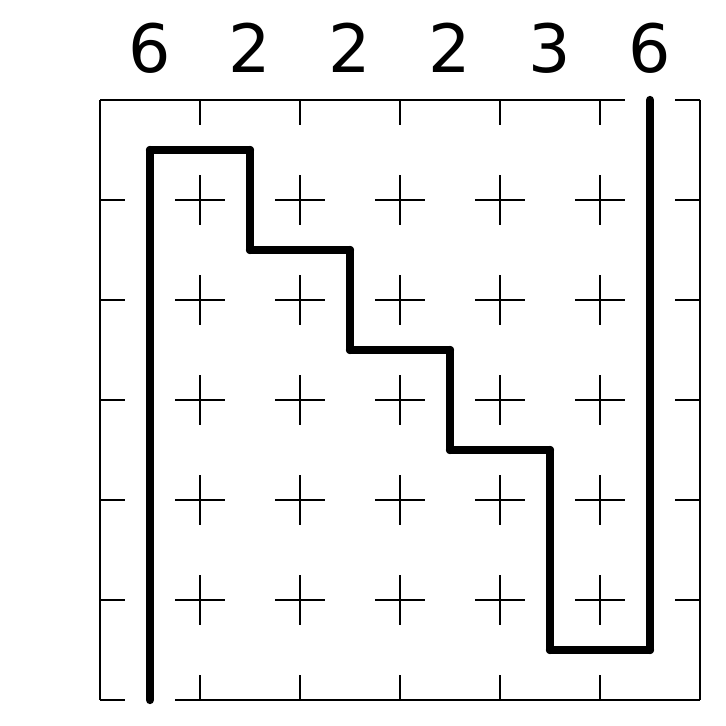}

  \includegraphics[width=.13\linewidth]{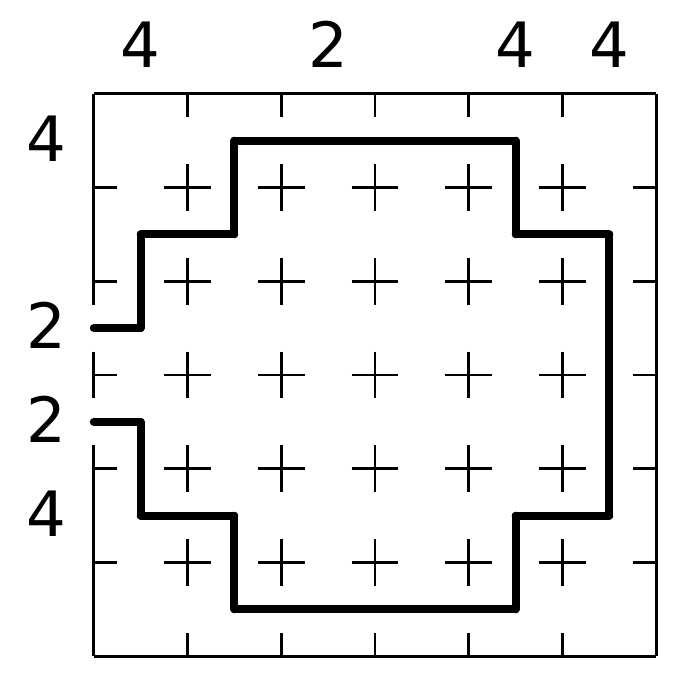}
  \includegraphics[width=.13\linewidth]{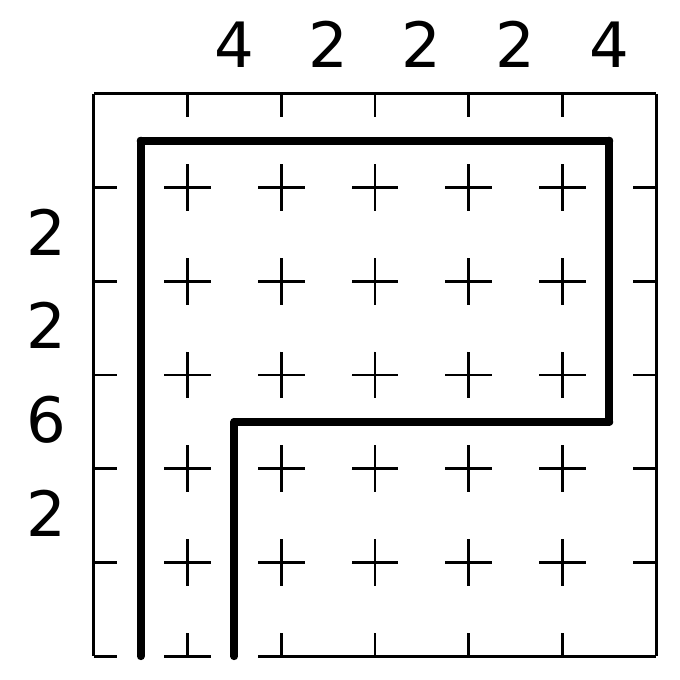}
  \includegraphics[width=.13\linewidth]{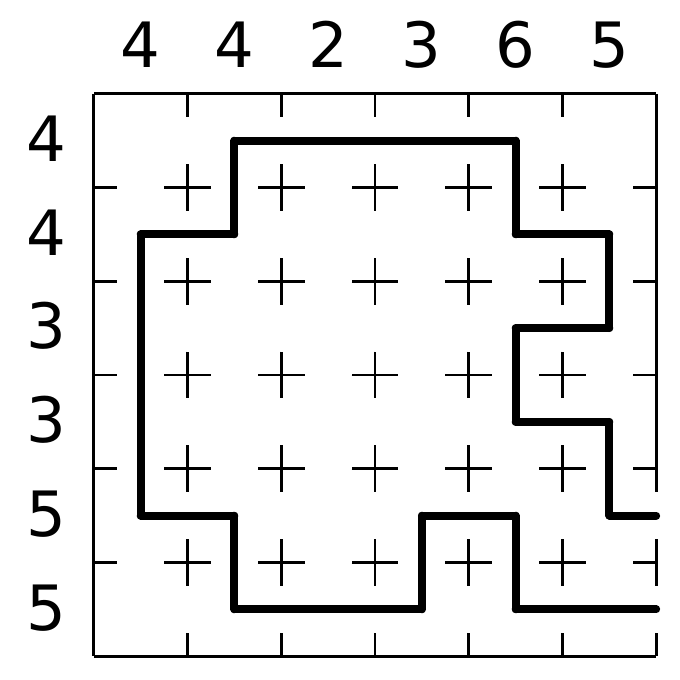}
  \includegraphics[width=.13\linewidth]{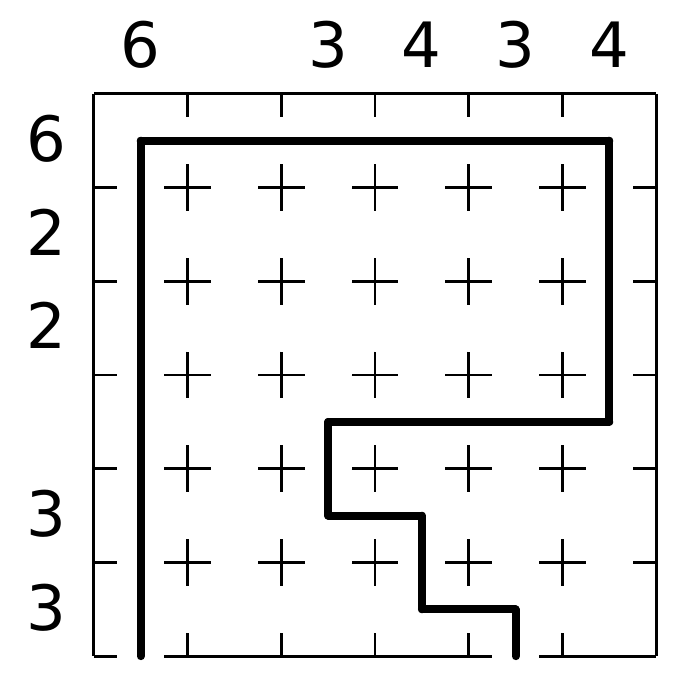}
  \includegraphics[width=.13\linewidth]{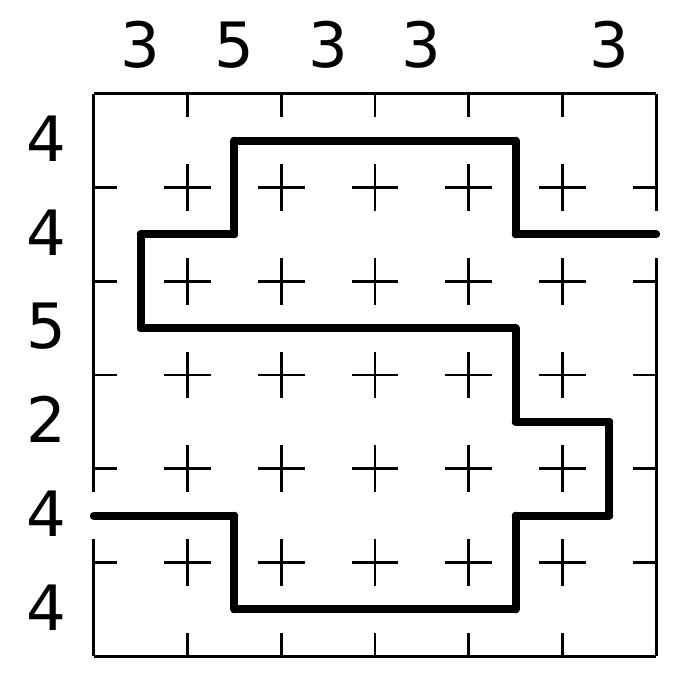}
  \includegraphics[width=.13\linewidth]{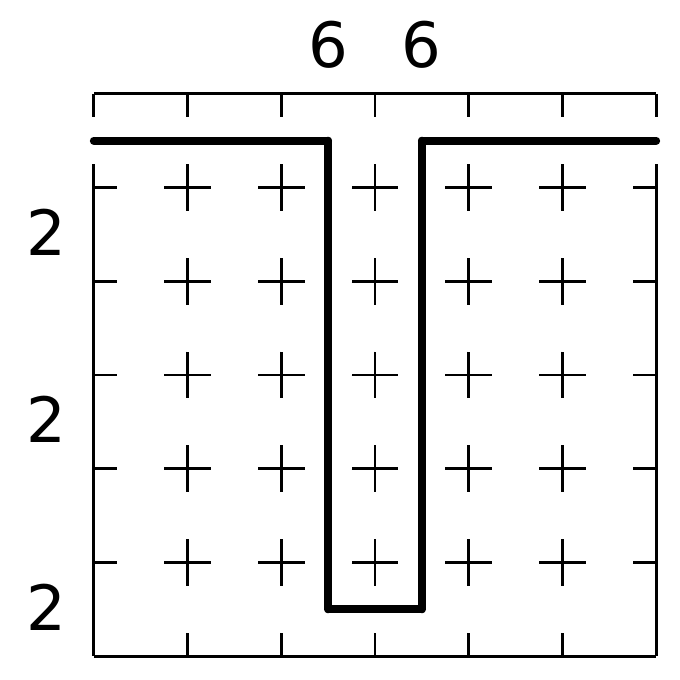}
  \hspace{.13\linewidth}

  \includegraphics[width=.13\linewidth]{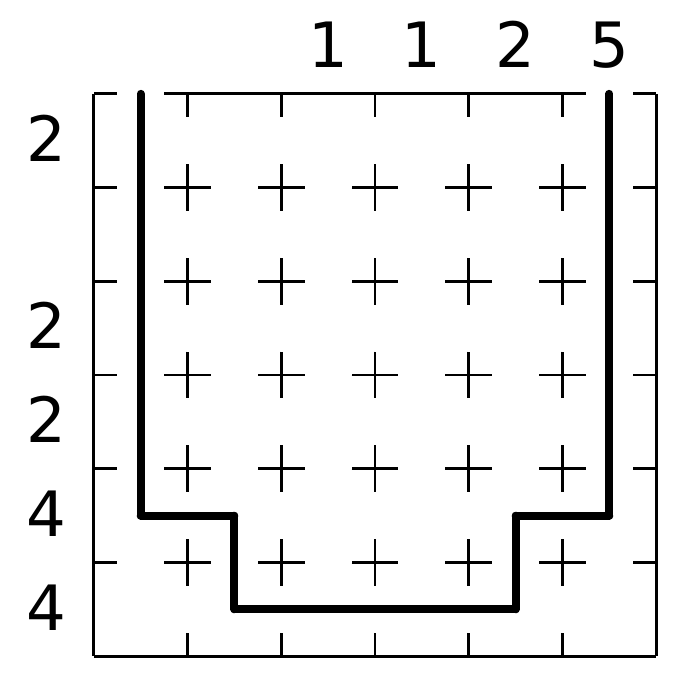}
  \includegraphics[width=.13\linewidth]{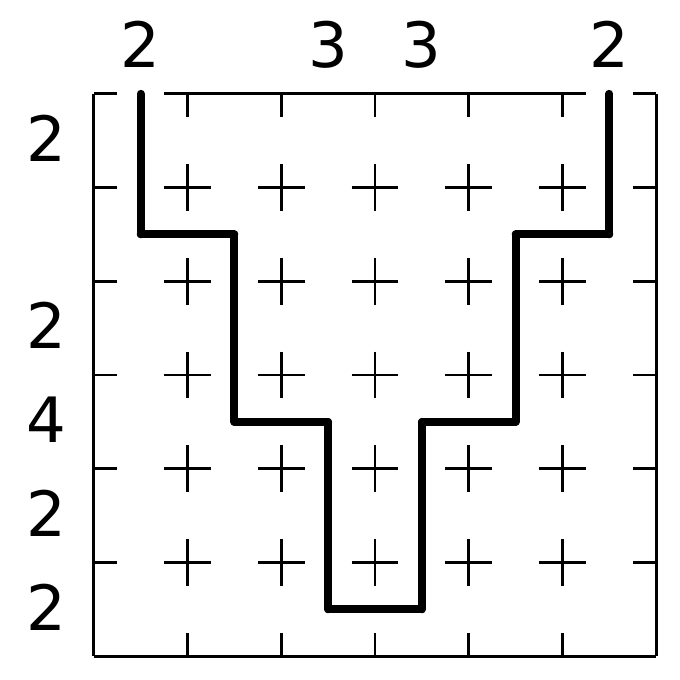}
  \includegraphics[width=.13\linewidth]{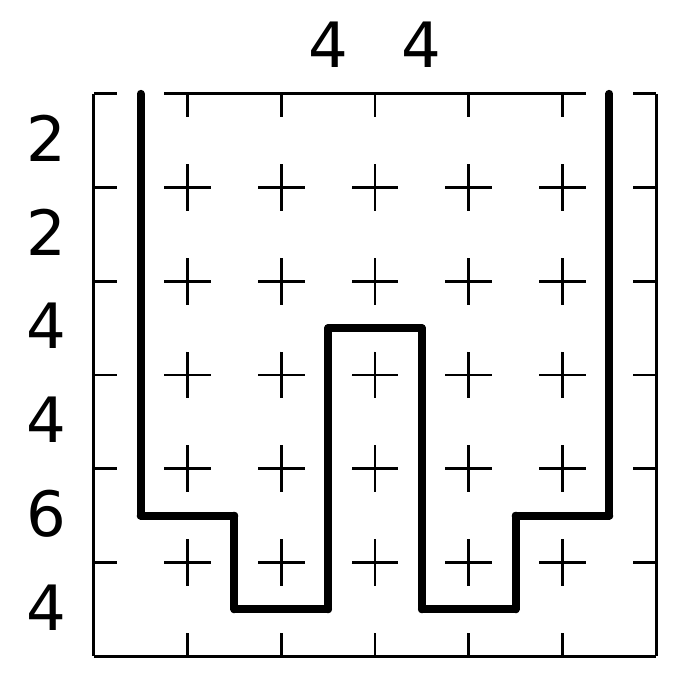}
  \includegraphics[width=.13\linewidth]{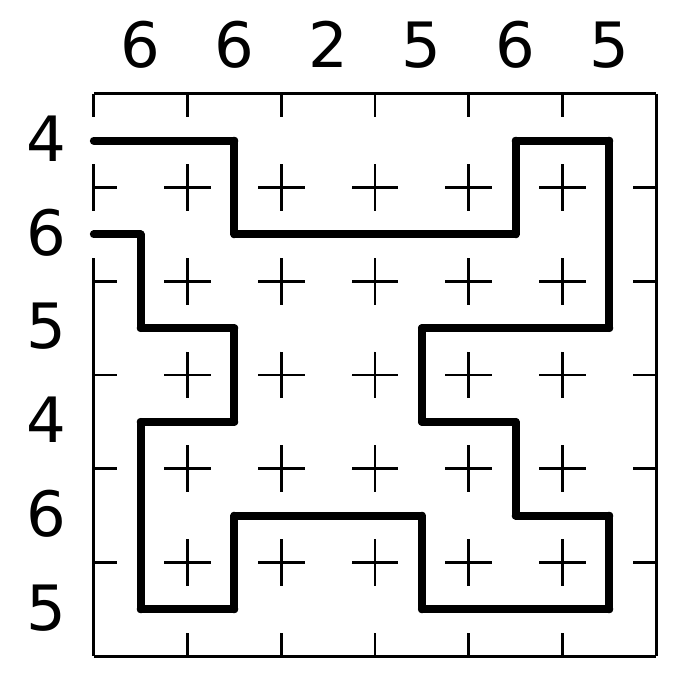}
  \includegraphics[width=.13\linewidth]{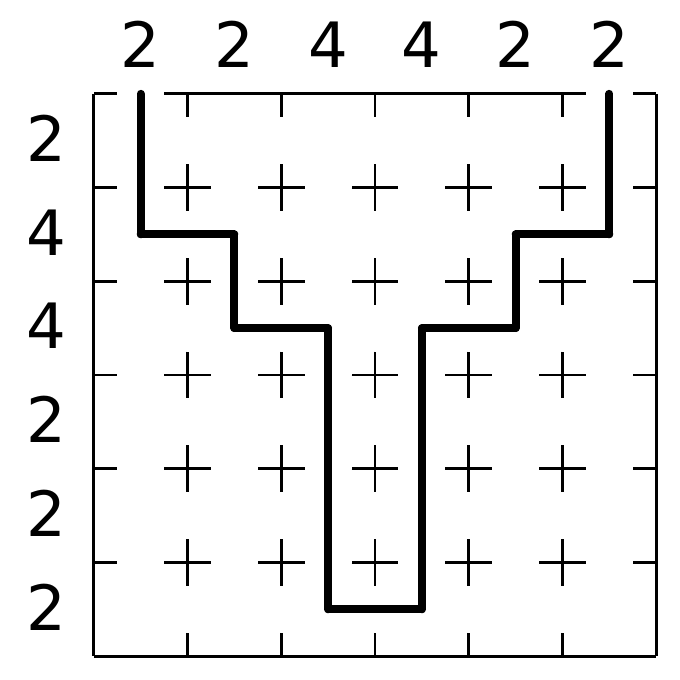}
  \includegraphics[width=.13\linewidth]{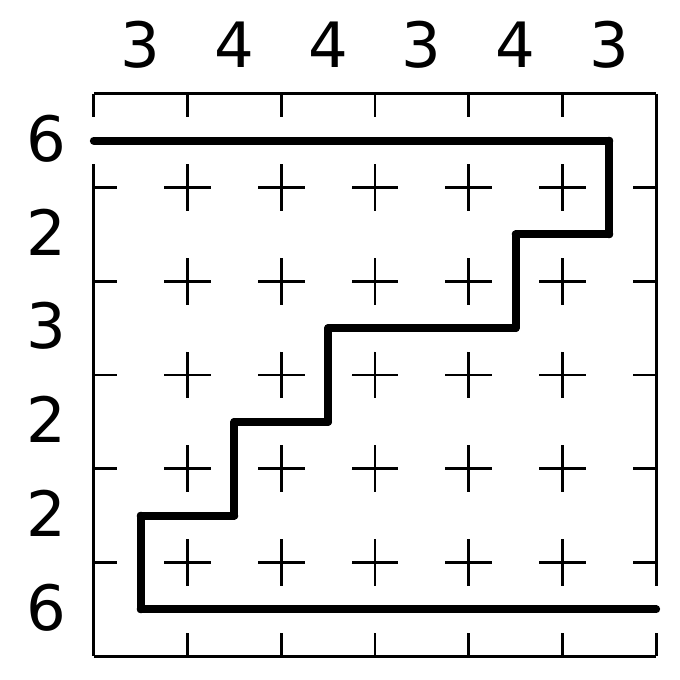}
  \hspace{.13\linewidth}
  \caption{Solved font}
  \label{fig:font-solved}
\end{figure}

\end{document}